\documentclass[12pt]{article}

\usepackage{amsmath}
\usepackage{amsthm, amssymb, amscd, amsxtra, amsfonts}

\usepackage[margin = 3cm]{geometry}

\usepackage[pdftex]{hyperref} %this has to be the last package
\hypersetup{pdfstartview=FitH,
	pdfpagemode=UseNone}

\newcommand{\field}[1]{\mathbb{#1}}
\newcommand{\N}{\field{N}}
\newcommand{\Z}{\field{Z}}

\newcommand{\R}{\field{R}}

\newcommand{\F}{\field{F}}
\newcommand{\PP}{\field{P}}
\newcommand{\term}[1]{\boldsymbol{#1}}
\newcommand{\sps}{\mathrm{\Sigma}\mathrm{\Pi}\mathrm{\Sigma}}
\newcommand{\spsp}{\mathrm{\Sigma}\mathrm{\Pi}\mathrm{\Sigma}\mathrm{\Pi}}

\DeclareMathOperator{\ch}{ch}
\DeclareMathOperator{\rk}{rk}
\DeclareMathOperator{\poly}{poly}
\DeclareMathOperator{\trdeg}{trdeg}
\DeclareMathOperator{\size}{size}
\DeclareMathOperator{\Sp}{\mathcal{S}}%set of sparse polys
\DeclareMathOperator{\res}{res}%resultant
\DeclareMathOperator{\simple}{sim}
\DeclareMathOperator{\sgn}{sgn}
\DeclareMathOperator{\id}{id}%identity permutation
\DeclareMathOperator{\sparse}{sp}%sparsity

\providecommand{\abs}[1]{\left\lvert#1\right\rvert}
\providecommand{\ol}[1]{\overline{#1}}

\theoremstyle{plain}
\newtheorem{theorem}{Theorem}

\newtheorem{corollary}[theorem]{Corollary}
\newtheorem{lemma}[theorem]{Lemma}
\newtheorem{conj}[theorem]{Conjecture}

\theoremstyle{definition}
\newtheorem{defn}[theorem]{Definition}
\newtheorem*{remark}{Remark}

\begin{document}

\title{Algebraic Independence and Blackbox Identity Testing}
\author{Malte Beecken \qquad Johannes Mittmann \qquad Nitin Saxena}
\date{Hausdorff Center for Mathematics, Bonn, Germany\\
	\small\texttt{\{malte.beecken, johannes.mittmann, nitin.saxena\}@hcm.uni-bonn.de}}

\maketitle

%+++++++++++++++++++++++++++++++++++++++++++++++++++++++++++++++++++++++++++++++++++++++++++++++++++++
%Abstract
%+++++++++++++++++++++++++++++++++++++++++++++++++++++++++++++++++++++++++++++++++++++++++++++++++++++

\begin{abstract}
	Algebraic independence is an advanced notion in commutative algebra that generalizes independence of 
	linear polynomials to higher degree. Polynomials $\{f_1,\ldots,f_m\} \subset \F[x_1,\ldots,x_n]$ 
	are called algebraically
	independent if there is no non-zero polynomial $F$ such that $F(f_1,\ldots,f_m)=0$. The transcendence
	degree, $\trdeg\{f_1,\ldots,f_m\}$, is the maximal number $r$ of algebraically independent polynomials 
	in the set. In this paper we
	design blackbox and efficient linear maps $\varphi$ that reduce the number of variables from $n$ to $r$ 
	but maintain $\trdeg\{\varphi(f_i)\}_i=r$, assuming $f_i$'s sparse and small $r$. We apply these fundamental
	maps to solve several cases of blackbox identity testing:
	\begin{enumerate}
 	  \item Given a polynomial-degree circuit $C$ and sparse polynomials $f_1,\ldots,f_m$ with $\trdeg$ $r$,
 	  we can test blackbox $D:=C(f_1,\ldots,f_m)$ for zeroness in $\poly(\size(D))^r$ time.
	  \item Define a $\spsp_\delta(k,s,n)$ circuit $C$ to be of the form $\sum_{i=1}^k\prod_{j=1}^s f_{i,j}$, 
	  where $f_{i,j}$ are sparse $n$-variate polynomials of degree at most $\delta$. For $k=2$ we give a 
	  $\poly(\delta sn)^{\delta^2}$ time blackbox identity test.
	  \item For a general depth-$4$ circuit we define a notion of rank. Assuming there is a 
	  rank bound $R$ for minimal simple $\spsp_\delta(k,s,n)$ identities, we give a $\poly(\delta snR)^{Rk\delta^2}$ 
	  time blackbox identity test for $\spsp_\delta(k,s,n)$ circuits. This partially generalizes the 
	  state of the art of depth-$3$ to depth-$4$ circuits.  
	\end{enumerate}
	The notion of $\trdeg$ works best with large or zero characteristic, but we also give versions of our
	results for arbitrary fields.
	
\end{abstract}

\newpage
%\setcounter{page}{1}

%+++++++++++++++++++++++++++++++++++++++++++++++++++++++++++++++++++++++++++++++++++++++++++++++++++++
%Main document
%+++++++++++++++++++++++++++++++++++++++++++++++++++++++++++++++++++++++++++++++++++++++++++++++++++++

\section{Introduction}

Polynomial identity testing (PIT) is the problem of checking whether a given $n$-variate arithmetic circuit 
computes the zero polynomial in $\F[x_1,\ldots,x_n]$. It is a central question in complexity theory as circuits
model computation and PIT leads us to a better understanding of circuits. There are several classical randomized
algorithms known \cite{bib:DL78,Sch80,Z79,CK00,LV98,AB03} that solve PIT. The basic Schwartz-Zippel test is: given a
circuit $C(x_1,\ldots,x_n)$, check $C(\ol{a})=0$ for a random $\ol{a}\in\ol{\F}^n$. Finding a 
deterministic polynomial time test, however, has been more difficult and is currently open. Derandomization 
of PIT is well 
motivated by a host of algorithmic applications, eg. bipartite matching \cite{L79} and matrix completion
\cite{L89}, and connections to 
sought-after super-polynomial lower bounds \cite{HS80,KI04}. Especially, {\em blackbox} PIT (i.e. circuit $C$ 
is given as a blackbox and we could only make oracle queries) has direct connections to lower bounds for the 
permanent \cite{A05,A06}. Clearly, finding a blackbox PIT test for a family of circuits $\mathcal{F}$ boils 
down to efficiently
designing a {\em hitting set} $\mathcal{H}\subset\ol{\F}^n$ such that: given a nonzero $C\in\mathcal{F}$,
there exists an $\ol{a}\in\mathcal{H}$ that {\em hits} $C$, i.e. $C(\ol{a})\ne0$. 

The attempts to solve blackbox PIT have focused on restricted circuit families. A natural restriction is 
{\em constant depth}. Agrawal \& Vinay \cite{AV08} showed that a blackbox PIT algorithm for depth-$4$ circuits 
would (almost) solve PIT for general circuits (and prove exponential circuit lower bounds for permanent). 
The currently known blackbox PIT algorithms work only for further restricted depth-$3$ and depth-$4$
circuits. The case of {\em bounded top fanin} depth-$3$ circuits has received great attention and has blackbox 
PIT algorithms \cite{DS06,KS07,KSh08,bib:SS10,KSar09,SS10focs,SS11stoc}. The analogous case for depth-$4$ 
circuits is open. However, with the additional restriction of {\em multilinearity} on all the multiplication 
gates, there is a blackbox PIT algorithm \cite{KMSV10,SV11}. The latter is somewhat subsumed by the PIT
algorithms for constant-read multilinear formulas \cite{AvMV10}. To save space we would not go into the 
rich history of PIT and instead refer to the surveys \cite{S09,bib:SY10}.

A recurring theme in the blackbox PIT research on depth-$3$ circuits has been that of {\em rank}. If we consider
a $\sps(k,d,n)$ circuit $C=$ $\sum_{i=1}^k\prod_{j=1}^d\ell_{i,j}$, where $\ell_{i,j}$ are linear forms in
$\F[x_1,\ldots,x_n]$, then $\rk(C)$ is defined to be the linear rank of the set of forms 
$\{\ell_{i,j}\}_{i,j}$ each viewed
as a vector in $\F^n$. This raises the natural question: Is there a generalized notion of rank for depth-$4$
circuits as well, and more importantly, one that is useful in blackbox PIT? We answer this question affirmatively
in this paper.
Our notion of rank is via {\em transcendence degree} (short, $\trdeg$), which is a basic notion in 
commutative algebra. To show that this notion applies to PIT requires relatively advanced algebra and new
tools that we build. 

Consider polynomials $\{f_1,\ldots,f_m\}$ in $\F[x_1,\ldots,x_n]$. They are called {\em algebraically
independent} (over $\F$) if there is no nonzero polynomial $F\in\F[y_1,\ldots,y_m]$ such that $F(f_1,\ldots,f_m)=0$. 
When those polynomials are {\em algebraically dependent} then such an $F$ exists and is called the 
{\em annihilating polynomial} of $f_1,\ldots,f_m$. 
The {\em transcendence degree}, $\trdeg\{f_1,\ldots,f_m\}$, is the maximal number $r$ of 
algebraically independent polynomials in the set $\{f_1,\ldots,f_m\}$. Though intuitive, it is nontrivial to
prove that $r$ is at most $n$ \cite{bib:Mo96}. The notion of $\trdeg$ has appeared in complexity 
theory in 
several contexts. Kalorkoti \cite{Kal85} used $\trdeg$ to prove an $\mathrm\Omega(n^3)$ formula size lower bound for
$n\times n$ determinant. In the works \cite{bib:DGW09,DGRV11} studying the {\em entropy} of polynomial mappings
$(f_1,\ldots,f_m):$ $\F^n\rightarrow\F^m$, $\trdeg$ is a natural measure of entropy when the field has large
or zero characteristic. It also appears implicitly in \cite{Dvir09} while constructing {\em extractors} for
varieties. Finally, the complexity of the annihilating polynomial is studied in \cite{bib:Kay09}. However, our 
work is the first to study $\trdeg$ in the context of PIT.

\subsection{Our main results}\label{sec:mainResults}

%In Sect. \ref{sec:FaithfulHomomorphisms} we design generic linear maps that we call {\em faithful homomorphisms}.
%They possess the useful property of preserving the trdeg of a given set of sparse polynomials 
Our first result shows that a general arithmetic circuit is sensitive to the $\trdeg$ of its input.

\begin{theorem}\label{thm:main1}
Let $C$ be an $m$-variate circuit. Let $f_1,\ldots,f_m$ be 
$\ell$-sparse, $\delta$-degree, $n$-variate polynomials with $trdeg$ $r$. Suppose we have oracle access to 
the $n$-variate $d$-degree circuit $C':=C(f_1,\ldots,f_m)$. There is a blackbox 
$\poly(\size(C')\cdot d\ell\delta)^r$ time test 
to check $C'=0$ (assuming a zero or larger than $\delta^r$ characteristic).  
\end{theorem}

We also give an algorithm that works for all fields but has a worse time complexity. Note that the above theorem
seems nontrivial even for a constant $m$, say $C'=C(f_1,f_2,f_3)$, as the output of $C'$ may not be sparse and 
$f_i$'s are of arbitrary degree and arity. In such a case $r$ is constant too and the theorem gives a polynomial
time test. Another example, where $r$ is constant but both $m$ and $n$ are variable, is: 
$f_i:=(x_1^i + x_2^2 + \dotsb + x_n^2)x_n^i$ for $i\in[m]$. (Hint: $r \le 3$.) 

Our next two main results concern depth-$4$ circuits. By $\spsp_{\delta}(k,s,n)$ we denote 
circuits (over a field $\F$) of the form 
\begin{equation}\label{eqn:spsp}
C:=\sum_{i=1}^k\prod_{j=1}^s f_{i,j},
\end{equation} 
where $f_{i,j}$'s are sparse $n$-variate polynomials of maximal degree $\delta$. Note that when $\delta=1$ this 
notation agrees with that of a $\sps$ circuit. Currently, the PIT methods are not
even strong enough to study $\spsp_{\delta}(k,s,n)$ circuits with both {\em top} fanin $k$ and {\em bottom}
fanin $\delta$ {\em bounded}. It is in this spectrum that we make exciting progress.

\begin{theorem}\label{thm:main2}
Let $C$ be a $\spsp_{\delta}(2,s,n)$ circuit over an arbitrary field. There is a blackbox 
$\poly(\delta sn)^{\delta^2}$ time test to check $C=0$.  
\end{theorem}

\paragraph{Simple, minimal and rank}
Finally, we define a notion of rank for depth-$4$ circuits and show its usefulness.
For a circuit $C$, as in (\ref{eqn:spsp}),
we define its {\em rank}, $\rk(C):=$ $\trdeg\{f_{i,j}\mid i\in[k], j\in[s]\}$. Define 
$T_i:=$ $\prod_{j=1}^s f_{i,j}$, for all $i\in[k]$, to be the {\em multiplication terms} of $C$. We call 
$C$ {\em simple} if $\{T_i \,\vert\; i\in[k]\}$ are coprime polynomials. We call $C$ {\em minimal} if there is 
no $I \subsetneq [k]$ such that $\sum_{i \in I}T_i=0$. Define $R_\delta(k,s)$ to be the smallest $r$ such that:
any $\spsp_\delta(k,s,n)$ circuit $C$ that is simple, minimal and zero has $\rk(C)<r$.

\begin{theorem}\label{thm:main3}
Let $r:=R_\delta(k,s)$ and the characteristic be zero or larger than $\delta^r$. There is a blackbox
$\poly(\delta rsn)^{rk\delta^2}$ time identity test for $\spsp_\delta(k,s,n)$ circuits.
\end{theorem}

We give a lower bound of $\mathrm\Omega(\delta k\log s)$ on $R_\delta(k,s)$ and conjecture an upper bound
(better than the trivial $ks$).

\subsection{Organization and our approach}

A priori it is not clear whether the problem of deciding algebraic independence of given polynomials
$\{f_1,\ldots,f_m\}$, over a field $\F$, is even computable. Perron \cite{bib:Per27} proved that for $m=(n+1)$
and any field, the annihilating polynomial has degree only exponential in $n$. We generalize this to any $m$ 
in Sect. \ref{sec:AlgebraicIndependence}, hence, deciding algebraic independence (over any field) is 
computable. When the characteristic is zero or large, there is a more efficient criterion due to Jacobi (Sect.
\ref{sec:JacobianCriterion}). For using $\trdeg$ in PIT we would need to relate it to the {\em Krull dimension}
of algebras (Sect. \ref{sec:AffineAlgebrasKrullDimension}). 

The central concept that we develop is that of a {\em faithful homomorphism}. This is a linear map $\varphi$ 
from $R:=\F[x_1,\ldots,x_n]$ to $\F[z_1,\ldots,z_r]$ such that for polynomials $f_1,\ldots,f_m\in R$ of
$\trdeg$ $r$, the images $\varphi(f_1),\ldots,\varphi(f_m)$ are also of $\trdeg$ $r$. Additionally, to be useful,
$\varphi$ should be constructible in a blackbox and efficient way. We give such constructions in Sects. 
\ref{sec:FaithfulHomomorphismExistence} and \ref{sec:VandermondeHomomorphism}. The proofs here use 
Perron's and Jacobi's criterion, but require new techniques as well. The reason why such a $\varphi$ is useful in PIT
is because it preserves the nonzeroness of the circuit $C(f_1,\ldots,f_m)$ (Corollary
\ref{cor:FaithfulHomomorphismOnCkt}). We prove this by an elegant application of Krull's {\em principal 
ideal theorem}.
 
Once the fundamental machinery is set up, we prove Theorem \ref{thm:main1} by designing a hitting set. The
zero or large characteristic case is handled in Sect. \ref{sec:CircuitsWithSparseSubcircuitsCharZero}. The 
arbitrary characteristic case is in Sect. \ref{sec:CircuitsWithSparseSubcircuitsArbitraryChar}. 

Finally, we apply the faithful homomorphisms to depth-$4$ circuits. The proof of Theorem \ref{thm:main2} is
provided in Sect. \ref{sec:PreservingSimplePart}. The rank-based hitting set is constructed in Sect. 
\ref{sec:Depth4CircuitsHittingSet} proving Theorem \ref{thm:main3}. 
The full proofs tend to be extremely technical and have been moved to the appendix.

\section{Preliminaries: Perron, Jacobi \& Krull} \label{sec:Preliminaries}

Let $n\in\Z^+$ and let $K$ be a field of characteristic $\ch(K)$. Throughout this paper, 
$K[\term{x}] = K[x_1, \dotsc, x_n]$ is a polynomial ring in $n$ variables over $K$.
$\ol{K}$ denotes the {\em algebraic closure} of the field. We denote the multiplicative 
{\em group of units} of an algebra $A$ by $A^*$.
We use the notation $[n] := \{1, \dotsc, n\}$. For $0 \le r \le n$, $\tbinom{[n]}{r}$ denotes
the set of $r$-subsets of $[n]$.

\subsection{Perron's criterion (arbitrary field)} \label{sec:AlgebraicIndependence}

Let $f_1, \dotsc, f_m \in K[\term{x}]$ be polynomials. When we want to emphasize the base field with the 
transcendence degree, we would use the notation $\trdeg_K\{f_1, \dotsc, f_m\}$. It is interesting to
note that transcendence degree is invariant to {\em algebraic} field extensions, i.e. $\trdeg_K\{f_1, \dotsc, f_m\}$ is 
the same as $\trdeg_{\ol{K}}\{f_1, \dotsc, f_m\}$ (Lemma \ref{lem:AnnihilatingPolynomialFieldExtension}).
The name transcendence degree stems from
field theory. The transcendence degree of a field extension $L/K$, denoted by $\trdeg(L/K)$, is the cardinality 
of any transcendence basis for $L/K$ (for more information on transcendental extensions, see 
\cite[Chap. 19]{bib:Mo96}).
For $L = K(f_1, \dotsc, f_m)$, we have $\trdeg_K\{f_1, \dotsc, f_m\}$ $= \trdeg(L/K)$ 
(cf. \cite[Theorem 19.14]{bib:Mo96}). Since $\trdeg(K(\term{x})/K) = n$, we obtain $0 \le$ 
$\trdeg_K\{f_1, \dotsc, f_m\}$ $\le n$.

Algebraic independence over $K$ strongly resembles $K$-linear independence. In fact, algebraic independence 
makes a finite subset $\{f_1, \dotsc, f_m\} \subset K[\term{x}]$ into a {\em matroid} (a generalization of
vector space, cf. \cite[Sect. 6.7]{bib:Oxl06}).

An effective criterion for algebraic independence can be obtained by a degree bound for annihilating polynomials.
The following theorem provides such a bound for the case of $n+1$ polynomials in $n$ variables.

\begin{theorem}[Perron's theorem]{\em \cite[Theorem 1.1]{bib:Plo05}}\label{thm:PerronsTheorem}
	Let $f_i \in K[\term{x}]$ be a polynomial of degree
	$\delta_i\ge1$, for $i\in[n+1]$. Then there exists a non-zero polynomial 
	$F \in K[y_1, \dotsc, y_{n+1}]$ such that $F(f_1, \dotsc, f_{n+1}) = 0$ and
	$\deg(F) \le$ $(\prod_i\delta_i)/\min_i\{\delta_i\}$. 
\end{theorem}

In the following corollary we give a degree bound in the general situation, where more variables than polynomials
are allowed. Moreover, the bound is in terms of the $\trdeg$ of the polynomials instead of the number
of variables. We hereby improve \cite[Theorem 11]{bib:Kay09} and generalize it to arbitrary characteristic. 
The proof uses a result
from Sect. \ref{sec:FaithfulHomomorphisms} and is given in Appendix \ref{app:AlgebraicIndependence}.

\begin{corollary}[Degree bound for annihilating polynomials] \label{cor:AnnihilatingPolynomialDegreeBound}
	Let $f_1, \dotsc, f_m$ $\in K[\term{x}]$ be algebraically dependent polynomials of 
	maximal degree $\delta$ and $\trdeg$ $r$.
	Then there exists a non-zero polynomial $F \in K[y_1, \dotsc, y_m]$ of degree at most $\delta^r$ such that
	$F(f_1, \dotsc, f_m) = 0$.
\end{corollary}
\begin{proof}[Proof sketch]
In Lemma \ref{lem:FaithfulHomomorphismExistence} we construct a homomorphism (by first principles) that 
reduces the number of 
variables to $r$ and preserves the $\trdeg$. We can then invoke Perron's theorem on $r+1$ of the polynomials.
\end{proof}

\begin{remark}
	The bound in Corollary \ref{cor:AnnihilatingPolynomialDegreeBound} is tight. To see this, let $n \ge 2$,
	let $\delta \ge 1$ and define the polynomials, $f_1 := x_1$, $f_2 := x_2 - x_1^{\delta}$, $\dotsc$, 
	$f_n := x_n - x_{n-1}^{\delta}$, $f_{n+1} := x_n^{\delta}$ in $K[\term{x}]$. Then 
	$\trdeg\{f_1, \dotsc, f_{n+1}\} = n$ and every annihilating polynomial
	of $f_1, \dotsc, f_{n+1}$ has degree at least $\delta^n$.
\end{remark}

\subsection{Jacobi's criterion (large or zero characteristic)} \label{sec:JacobianCriterion}

In large or zero characteristic, the well-known Jacobian criterion yields a more efficient criterion 
for algebraic independence. 

For $i \in [n]$, we denote the $i$-th formal partial derivative of a polynomial $f \in K[\term{x}]$
by $\partial_{x_i} f$. Now let $f_1, \dotsc, f_m \in K[\term{x}]$. Then
\[
	J_{\term{x}}(f_1, \dotsc, f_m) := \bigl(\partial_{x_j} f_i\bigr)_{i, j} = 
	\begin{pmatrix}
		\partial_{x_1} f_1 & \cdots & \partial_{x_n} f_1 \\
		\vdots & & \vdots \\
		\partial_{x_1} f_m & \cdots & \partial_{x_n} f_m
	\end{pmatrix} \in K[\term{x}]^{m \times n}
\]
is called the {\em Jacobian matrix} of $f_1, \dotsc, f_m$. Its matrix-rank over the function field is of
great interest.

\begin{theorem}[Jacobian criterion] \label{thm:JacobianCriterion}
	Let $f_1, \dotsc, f_m \in K[\term{x}]$ be polynomials of degree at most $\delta$ and $\trdeg$ $r$. 
	Assume that $\ch(K) = 0$ or $\ch(K) > \delta^r$. Then $\rk_{L} J_{\term{x}}(f_1, \dotsc, f_m)=
	\trdeg_K\{f_1, \dotsc,f_m\}$, where $L = K(\term{x})$.
\end{theorem}

A proof of the Jacobian criterion in characteristic $0$ appears, for example, in \cite{bib:ER93} and the case of large prime
characteristic was dealt with in \cite{bib:DGW09}. By virtue of Theorem \ref{thm:PerronsTheorem} 
our proof could tolerate a slightly smaller characteristic. For the reader's convenience, 
a full proof is given in Appendix \ref{app:JacobianCriterion}.
We isolate the following special case of Theorem \ref{thm:JacobianCriterion},
because it holds in arbitrary characteristic.

\begin{lemma} \label{lem:JacobianCriterion}
	Let $f_1, \dotsc, f_m \in K[\term{x}]$. Then $\trdeg_K\{f_1, \dotsc,f_m\}\ge$ 
	$\rk_{L} J_{\term{x}}(f_1, \dotsc$, $f_m)$, where $L = K(\term{x})$.
\end{lemma}

\subsection{Krull dimension of affine algebras} \label{sec:AffineAlgebrasKrullDimension}

In this section, we want to highlight the connection between transcendence degree and
the Krull dimension of affine algebras. This will enable us to use Krull's principal
ideal theorem which is stated below.

In this paper, a {\em $K$-algebra} $A$ is always a commutative ring containing $K$
as a subring. %(the usual definition includes the possibility $A = \{0\}$). 
The most important example of a $K$-algebra is $K[\term{x}]$. Let $A, B$ be $K$-algebras. 
A map $A \rightarrow B$ is called a {\em $K$-algebra homomorphism} if it is a ring homomorphism
that fixes $K$ element-wise. 

We want to extend the definition of algebraic independence to algebras (whose elements may not be the usual
polynomials any more). Let $a_1, \dotsc, a_m \in A$
and consider the $K$-algebra homomorphism
\[
	\rho: K[\term{y}] \rightarrow A, \qquad F \mapsto F(a_1, \dotsc, a_m) ,
\]
where $K[\term{y}] = K[y_1, \dotsc, y_m]$. If $\ker(\rho) = \{0\}$, then $\{a_1, \dotsc, a_m\}$
is called algebraically independent over $K$. If $\ker(\rho) \neq \{0\}$, 
then $\{a_1, \dotsc, a_m\}$ is called algebraically dependent over $K$.
For a subset $S \subseteq A$, we define the transcendence degree of $S$ over $K$ by an obvious
supremum:
\[
	\trdeg_K(S) := \sup\bigl\{ \abs{T	} \,\vert\; \text{$T \subseteq S$ is finite and algebraically independent}\bigr\} .
\]
The image of $K[\term{y}]$ under $\rho$ is the subalgebra of $A$ generated by $a_1, \dotsc, a_m$ and 
is denoted by
$K[a_1, \dotsc, a_m]$. An algebra of this form is called an {\em affine $K$-algebra}, and it
is called an {\em affine $K$-domain} if it is an integral domain.

The {\em Krull dimension} of $A$, denoted by $\dim(A)$, is defined as the supremum over all $r \ge 0$ for 
which
there is a chain $\mathfrak{p}_0 \subsetneq \mathfrak{p}_1 \subsetneq \dotsb \subsetneq \mathfrak{p}_r$
of prime ideals $\mathfrak{p}_i \subset A$. It measures how far $A$ is from a field.

\begin{theorem}[Dimension and trdeg] \label{thm:DimEqualsTrdeg}
	Let $A = K[a_1, \dotsc, a_m]$ be an affine $K$-algebra. Then $\dim(A) = \trdeg_K(A)$ 
	$= \trdeg_K\{a_1, \dotsc, a_m\}$.
\end{theorem}

\begin{proof}
	Cf. \cite[Theorem 5.9 and Proposition 5.10]{bib:Kem11}. Also, the integral domain case is in the standard
	text \cite[Theorem 5.6]{bib:Mat89}. 
\end{proof}

The following corollary is a simple consequence of Theorem \ref{thm:DimEqualsTrdeg}.
It shows that homomorphisms cannot increase the dimension of affine algebras.
The proof is given in Appendix \ref{app:AffineAlgebrasKrullDimension}.

\begin{corollary} \label{cor:DimNonIncreasing}
	Let $A,B$ be $K$-algebras and let $\varphi:A \rightarrow B$ be a 
	$K$-algebra homomorphism. If $A$ is an affine algebra, then so is $\varphi(A)$ and we have 
	$\dim(\varphi(A))$ $\le \dim(A)$. 
	If, in addition, $\varphi$ is injective, then $\dim(\varphi(A)) = \dim(A)$.
\end{corollary}

In the next section we will need the following version of Krull's principal ideal theorem.

\begin{theorem}[Krull's Hauptidealsatz] \label{thm:PrincipalIdealTheorem}
	Let $A$ be an affine $K$-domain and let $a \in A \setminus (A^* \cup \{0\})$.
	Then $\dim(A/\langle a \rangle) = \dim(A) - 1$.
\end{theorem}
\begin{proof}
	Cf. \cite[Corollary 13.11]{bib:Eis95} or \cite[Theorem 13.5]{bib:Mat89}. 
\end{proof}

\section{Faithful homomorphisms: Reducing the variables} \label{sec:FaithfulHomomorphisms}

Let $f_1, \dotsc, f_m \in K[\term{x}]$ be polynomials and let $r := \trdeg\{f_1, \dotsc, f_m\}$.
Intuitively, $r$ variables should suffice to define $f_1, \dotsc, f_m$ without changing
their algebraic relations.
So let $K[\term{z}] = K[z_1, \dotsc, z_r]$ be a polynomial ring with $1 \le r \le n$.
We want to find a homomorphism $K[\term{x}] \rightarrow K[\term{z}]$ that preserves
the transcendence degree of $f_1, \dotsc, f_m$. First we give this property a name.

\begin{defn}
	Let	$\varphi: K[\term{x}] \rightarrow K[\term{z}]$ be a $K$-algebra homomorphism.  
	We say $\varphi$ is {\em faithful to $\{f_1, \dotsc, f_m\}$} if
	$\trdeg\{\varphi(f_1), \dotsc, \varphi(f_m)\} = \trdeg\{f_1, \dotsc, f_m\}$.
\end{defn}

The following theorem shows that faithful homomorphisms are useful for us.

\begin{theorem}[Faithful is useful] \label{thm:FaithfulHomomorphismInjective}
	Let $A = K[f_1, \dotsc, f_m] \subseteq K[\term{x}]$. Then $\varphi$ is faithful to 
	$\{f_1, \dotsc, f_m\}$ if and only if $\varphi\vert_A: A \rightarrow K[\term{z}]$ is injective
	(iff $A\cong K[\varphi(f_1),\ldots,\varphi(f_m)]$).
\end{theorem}
\begin{proof}
	We denote $\varphi_A = \varphi\vert_A$ and $r = \trdeg\{f_1, \dotsc, f_m\}$. 
	If $\varphi_A$ is injective, then
	\[
		r = \dim(A) = \dim(\varphi_A(A))
		= \trdeg\{\varphi(f_1), \dotsc, \varphi(f_m)\}
	\]
	by Theorem \ref{thm:DimEqualsTrdeg} and Corollary \ref{cor:DimNonIncreasing}.
	Thus $\varphi$ is faithful to $\{f_1, \dotsc, f_m\}$.
	
	Conversely, let $\varphi$ be faithful to $\{f_1, \dotsc, f_m\}$. Then
	$\dim(\varphi_A(A)) = r$. Now assume for the sake of contradiction that
	$\varphi_A$ is not injective. Then there exists an $f \in A \setminus \{0\}$ such that
	$\varphi_A(f) = 0$. We have $f \notin K$, because $\varphi$ fixes $K$ element-wise, and
	hence $f \notin A^*$. Since $A$ is an affine domain, Theorem \ref{thm:PrincipalIdealTheorem} 
	implies $\dim(A / \langle f \rangle) = r-1$.
	Since $f \in \ker(\varphi_A)$, the $K$-algebra homomorphism
	\[
		\overline{\varphi}_A: A / \langle f \rangle \rightarrow K[\term{z}], \qquad a + \langle f \rangle \mapsto \varphi_A(a)
	\]
	is well-defined and $\varphi_A$ factors as $\varphi_A = \overline{\varphi}_A \circ \eta$, where
	$\eta:A \rightarrow A / \langle f \rangle$ is the canonical surjection. But then Corollary \ref{cor:DimNonIncreasing} 
	implies
	\[
		r = \dim(\varphi_A(A)) = \dim(\ol{\varphi}_A(\eta(A))) \le \dim(\eta(A)) =
		\dim(A / \langle f \rangle) = r-1 ,
	\]
	a contradiction. It follows that $\varphi_A$ is injective. 
	
	When $\varphi_A$ is injective, clearly we have $A\cong \varphi_A(A)=$ $K[\varphi(f_1),\ldots,\varphi(f_m)]$.
\end{proof}

\begin{corollary}\label{cor:FaithfulHomomorphismOnCkt}
Let $C$ be an $m$-variate circuit over $K$. Let $\varphi$ be faithful to $\{f_1,\ldots,$ $f_m\}$ 
$\subset K[\term{x}]$.
Then, $C(f_1,\ldots,f_m)=0$ iff $C(\varphi(f_1),\ldots,\varphi(f_m))=0$.  
\end{corollary}
\begin{proof}
Note that $C(f_1,\ldots,f_m)$ resp. $C(\varphi(f_1),\ldots,\varphi(f_m))$ are elements in 
the algebras $K[f_1, \dotsc, f_m]$ resp. $K[\varphi(f_1), \dotsc, \varphi(f_m)]$. Since $\varphi$ is an 
isomorphism between these two algebras, the corollary is evident. 
\end{proof}

\subsection{A Kronecker-inspired map (arbitrary characteristic)} \label{sec:FaithfulHomomorphismExistence}

The following lemma shows that even {\em linear} faithful homomorphisms exist for all subsets of polynomials
(provided $K$ is large enough, for eg. move to $\ol{K}$ or a large enough field extension \cite{AL86}). It is a 
generalization of \cite[Claim 11.1]{bib:Kay09} to arbitrary
characteristic. The proof is given in Appendix \ref{app:FaithfulHomomorphismExistence}.

\begin{lemma}[Existence] \label{lem:FaithfulHomomorphismExistence}
	Let $K$ be an infinite field and let $f_1, \dotsc, f_m \in K[\term{x}]$ be polynomials of 
	$\trdeg$ $r$. Then there exists a linear $K$-algebra homomorphism 
	$\varphi: K[\term{x}] \rightarrow K[\term{z}]$ which is faithful to
	$\{f_1, \dotsc, f_m\}$.
\end{lemma}
\begin{proof}[Proof sketch]
We prove this by first principles.
The proof is by identifying $r$ variables from $\{x_1,\ldots,x_n\}$ that we leave {\em free} and the rest
$n-r$ variables we fix to generic elements from $K$. Using annihilating polynomials we could show that 
this map preserves the $\trdeg$. 
\end{proof}

Below we want to make this lemma effective. This will again be accomplished by substituting constants
for all but $r$ of the variables $x_1, \dotsc, x_n$. We define a parametrized homomorphism $\mathrm\Phi$
in three steps. First, we decide which variables we want to keep and map them to $z_1, \dotsc, z_r$.
To the remaining variables we apply a {\em Kronecker substitution} using a new variable $t$, i.e. we
map the $i$-th variable to $t^{D^i}$ (for a large $D$). In the second step, 
the exponents of $t$ will be reduced modulo some number. Finally, a single constant will be substituted for $t$.

Let $I = \{j_1, \dotsc, j_r\} \in \tbinom{[n]}{r}$ be an index set and let
$[n] \setminus I = \{j_{r+1}, \dotsc, j_n\}$ be its complement such that $j_1 < \dotsb < j_r$ and
$j_{r+1} < \dotsb < j_n$. Let $D \ge 2$ and define the $K$-algebra homomorphism
\[
	\mathrm\Phi_{I,D}: K[\term{x}] \rightarrow K[t,\term{z}], \qquad x_{j_i} \mapsto \begin{cases}
		z_i, & \text{for $i = 1, \dotsc, r$}, \\
		t^{D^{i-r}}, & \text{for $i = r+1, \dotsc, n$} .
	\end{cases}
\]
Now let $p \ge 1$. For an integer $a \in \Z$, we denote by $\lfloor a \rfloor_p$ the integer 
$b \in \Z$ satisfying $0 \le b < p$ and $a = b \pmod{p}$. We define the $K$-algebra homomorphism
\[
	\mathrm\Phi_{I,D,p}: K[\term{x}] \rightarrow K[t,\term{z}], \qquad x_{j_i} \mapsto \begin{cases}
		z_i, & \text{for $i = 1, \dotsc, r$}, \\
		t^{\lfloor D^{i-r} \rfloor_p}, & \text{for $i = r+1, \dotsc, n$} .
	\end{cases}
\]
Note that, for $f \in K[\term{x}]$, $\mathrm\Phi_{I,D,p}(f)$ is
a representative of the residue class $\mathrm\Phi_{I,D}(f) \pmod{\langle t^p-1 \rangle_{K[t, \term{z}]}}$.
Finally let $c \in\ol{K}$ and define the $\ol{K}$-algebra homomorphism
\[
	\mathrm\Phi_{I,D,p,c}:\ol{K}[\term{x}] \rightarrow \ol{K}[\term{z}], \qquad f \mapsto \bigl(\mathrm\Phi_{I,D,p}(f)\bigr)(c, \term{z}) .
\]
The following lemma bounds the number of bad choices for the parameters $p$ and $c$. It is proven in
Appendix \ref{app:FaithfulHomomorphismExistence}.

\begin{lemma}[$\mathrm\Phi$ is faithful] \label{lem:ProjectionHomomorphismFaithful}
	Let $f_1, \dotsc, f_m \in K[\term{x}]$ be polynomials of degree at most $\delta$ and 
	$\trdeg$ at most $r$. Let $D > \delta^{r+1}$.
	Then there exist an index set $I \in \tbinom{[n]}{r}$ and a prime $p\le(n+\delta^r)^{8\delta^{r+1}}(\log_2D)^2+1$
	such that any subset of $\ol{K}$ of size $\delta^r r p$ contains $c$ such that
	$\mathrm\Phi_{I,D,p,c}$ is faithful to $\{f_1, \dotsc, f_m\}$.
\end{lemma}
\begin{proof}[Proof sketch]
We identify a maximal $I \subseteq [n]$ such that for the field $L:=K(x_i \,\vert\; i \notin I)$, 
$\trdeg_L\{f_1,\ldots,f_m\}$ $=\trdeg_K\{f_1, \dotsc, f_m\}$. Now $x_i$, for $i\in I$, 
is algebraic over the field $L(f_1,\ldots,f_m)$. This gives us annihilating polynomials whose degrees we could
bound by Corollary \ref{cor:AnnihilatingPolynomialDegreeBound}, and hence their sparsities. By sparse PIT tricks
we get a bound on the `good' $p$ and $c$. 
\end{proof}

In large or zero characteristic, a more efficient version of this lemma can be given (for the same homomorphism
$\mathrm\Phi$).
The reason is that we can work with the Jacobian criterion instead of the degree bound for annihilating
polynomials. However, we omit the statement of this result here, because we can give a more holistic construction
in that case. This will be presented in the following section.

\subsection{A Vandermonde-inspired map (large or zero characteristic)} \label{sec:VandermondeHomomorphism}

To prove Theorem \ref{thm:main3}, we will need a homomorphism that is
faithful to several sets of polynomials simultaneously. The homomorphism $\mathrm\Phi$ constructed in the previous
section does not meet this requirement, because its definition depends on a {\em fixed} subset of the 
variables $x_1, \dotsc, x_n$.
In this section we will devise a construction, that treats the variables $x_1, \dotsc, x_n$ in a uniform manner.
It is inspired by the {\em Vandermonde matrix}, i.e. $((t^{ij}))_{i,j}$.	 

We define a parametrized homomorphism $\mathrm\Psi$ in three steps.
Let $K[\term{z}] = K[z_0, \dotsc, z_r]$, where $1 \le r \le n$.
Let $D_1, D_2 \ge 2$ and let $D = (D_1, D_2)$. Define the $K$-algebra
homomorphism
\[
	\mathrm\Psi_D:K[\term{x}] \rightarrow K[t, \term{z}], \qquad x_i \mapsto t^{D_1^i} + t^{D_2^i}z_0 + \sum_{j=1}^r t^{i(n+1)^j} z_j ,
\]
where $i = 1, \dotsc, n$. This map (linear in the $z$'s) should be thought of as a variable reduction from 
$n$ to $r+1$. The coefficients of $z_1,\ldots,z_r$ bear resemblance to a row of a Vandermonde matrix, while
that of $z_0$ (and the constant coefficient) resembles Kronecker substitution. This definition is carefully tuned
so that $\mathrm\Psi$ finally preserves both the $\trdeg$ (proven here) and $\gcd$ of polynomials (proven in
Sect. \ref{sec:PreservingSimplePart}). 

Next let $p \ge 1$ and define the $K$-algebra homomorphism
\[
	\mathrm\Psi_{D,p}:K[\term{x}] \rightarrow K[t, \term{z}], \qquad 
	x_i \mapsto t^{\lfloor D_1^i \rfloor_p} + t^{\lfloor D_2^i \rfloor_p}z_0 + \sum_{j=1}^r t^{\lfloor i(n+1)^j \rfloor_p} z_j ,
\]
where $i = 1, \dotsc, n$. Note that, for $f \in K[\term{x}]$, $\mathrm\Psi_{D,p}(f)$ is
a representative of the residue class $\mathrm\Psi_D(f) \pmod{\langle t^p-1 \rangle_{K[t, \term{z}]}}$.
Finally let $c \in \ol{K}$ and define the $\ol{K}$-algebra homomorphism
\[
	\mathrm\Psi_{D,p,c}:\ol{K}[\term{x}] \rightarrow \ol{K}[\term{z}], \qquad f \mapsto \bigl(\mathrm\Psi_{D,p}(f)\bigr)(c, \term{z}) .
\]
The following lemma bounds the number of bad choices for the parameters $p$ and $c$. 
The proof, which is given in Appendix \ref{app:VandermondeHomomorphism}, uses the Jacobian criterion, 
therefore the lemma has a restriction on $\ch(K)$.

\begin{lemma}[$\mathrm\Psi$ is faithful] \label{lem:VandermondeHomomorphismFaithful}
	Let $f_1, \dotsc, f_m \in K[\term{x}]$ be polynomials of sparsity at most $\ell$, degree at most $\delta$
	and $\trdeg$ at most $r$. Assume that $\ch(K)=0$ or $\ch(K) > \delta^r$. Let $D = (D_1, D_2)$ such that
	$D_1 \ge \max\{\delta r+1, (n+1)^{r+1}\}$ and $D_2 \ge 2$. 	
	Then there exists a prime $p\le (2nr\ell)^{2(r+1)} (\log_2 D_1)^2+1$ such that any subset of $\ol{K}$ of size 
  $\delta r p$ contains $c$ such that
	$\mathrm\Psi_{D,p,c}$ is faithful to $\{f_1, \dotsc, f_m\}$.
\end{lemma}
\begin{proof}[Proof sketch]
We study the action of $\mathrm\Psi_D$ on the Jacobian determinant. Because of the chain rule of partial
derivatives, this leads us to a product of two
determinants, which we expand using the Cauchy-Binet formula and estimate its sparsity. By sparse PIT tricks
we get a bound on the `good' $p$ and $c$.   
\end{proof}

By trying larger $p$ and $c$, we can find a $\mathrm\Psi$ that is faithful to several subsets
of polynomials simultaneously. This is an advantage of $\mathrm\Psi$ over $\mathrm\Phi$, in addition to
being more efficiently constructible.

\section{Circuits with sparse inputs of low transcendence degree (proving Theorem \ref{thm:main1})} 
\label{sec:CircuitsWithSparseSubcircuits}

We can now proceed with the first PIT application of faithful homomorphisms.
We consider arithmetic circuits of the form $C(f_1, \dotsc, f_m)$, where $C$ is
a circuit computing a polynomial in $K[\term{y}] = K[y_1, \dotsc, y_m]$ and 
$f_1, \dotsc, f_m$ are subcircuits computing polynomials in $K[\term{x}]$.
Thus, $C(f_1, \dotsc, f_m)$ computes a polynomial in the subalgebra
$K[f_1, \dotsc, f_m]$.

Let $C(f_1, \dotsc, f_m)$ be of maximal degree $d$, and let $f_1, \dotsc, f_m$
be of maximal degree $\delta$, maximal sparsity $\ell$ and maximal transcendence
degree $r$. First, we use a faithful homomorphism to transform $C(f_1, \dotsc, f_m)$
into an $r$-variate circuit.
Then, a hitting set for $r$-variate degree-$d$ polynomials is used, given by the following
version of the Schwartz-Zippel lemma.

\begin{lemma}[Schwartz-Zippel] \label{lem:CombinatorialNullstellensatz}
	Let $H \subset \ol{K}$ be a subset of size $d + 1$.
	Then $\mathcal{H} = H^r$ is a hitting set for 
	$\{ f \in K[z_1, \dotsc, z_r] \,\vert\; \deg(f) \le d \}$.
\end{lemma}
\begin{proof}
	Cf. \cite[Lemma 2.1]{bib:Alo99}. 
\end{proof}

\subsection{A hitting set (large or zero characteristic)} \label{sec:CircuitsWithSparseSubcircuitsCharZero}

We use the map $\mathrm\Psi$ from Sect. \ref{sec:VandermondeHomomorphism}.
This hitting set construction is efficient for $r$ constant and $\ell$, $d$ polynomial in the input size. 

Let $n, d,r,\delta,\ell \ge 1$ and let $K[\term{z}] = K[z_0, z_1, \dotsc, z_r]$. 
We introduce the following parameters.
\begin{enumerate}
	\item Define $D = (D_1, D_2)$ by $D_1 := (2\delta n)^{r+1}$ and $D_2 := 2$.
	\item Define $p_{\max} := (2nr\ell)^{2(r+1)} \lceil\log_2 D_1\rceil^2+1$.
	\item Pick arbitrary $H_1, H_2 \subset \ol{K}$ of sizes $\delta r p_{\max}$ resp. $d + 1$.
	%\item Let $H_2 \subset \ol{K}$ be an arbitrary subset of size $d + 1$.
\end{enumerate}
Denote $\mathrm\Psi_{D,p,c}^{(i)} := \mathrm\Psi_{D,p,c}(x_i) \in \ol{K}[\term{z}]$ for $i = 1, \dotsc, n$
and define the subset
\[
	\mathcal{H}_{d,r,\delta,\ell} = \Bigl\{ \bigl(\mathrm\Psi_{D,p,c}^{(1)}(\boldsymbol{a}), \dotsc, \mathrm\Psi_{D,p,c}^{(n)}(\boldsymbol{a})\bigr) 
	\,\bigl\vert\; \text{$p \in [p_{\max}]$, $c \in H_1$, $\boldsymbol{a} \in H_2^{r+1}$} \Bigr\} \subset \ol{K}^n .
\]
The following theorem shows that, over a large or zero characteristic, this is a 
hitting set for the class of circuits under consideration.
A proof is given in Appendix \ref{app:CircuitsWithSparseSubcircuitsCharZero}.

\begin{theorem} \label{thm:CircuitsWithSparseSubcircuitsHittingSetCharZero}
	Assume that $\ch(K) = 0$ or $\ch(K) > \delta^r$. Then $\mathcal{H}_{d,r,\delta,\ell}$ is
	a hitting set for the class of degree-$d$ circuits with inputs being $\ell$-sparse, degree-$\delta$ subcircuits
	of $\trdeg$ at most $r$. It can be constructed in $\poly(dr\delta\ell n)^r$ time.
\end{theorem}

\subsection{A hitting set (arbitrary characteristic)} \label{sec:CircuitsWithSparseSubcircuitsArbitraryChar}

We use the map $\mathrm\Phi$ from Sect. \ref{sec:FaithfulHomomorphismExistence}.
This hitting set construction is
efficient for $\delta$, $r$ constants and $d$ polynomial in the input size.

Let $n,d,r,\delta \ge 1$ and let $K[\term{z}] = K[z_1, \dotsc, z_r]$.
We introduce the following parameters.
\begin{enumerate}
	\item Define $D := \delta^{r+1}+1$.
	\item Define $p_{\max} := (n+\delta^r)^{8\delta^{r+1}}\lceil\log_2D\rceil^2+1$.
	\item Pick arbitrary $H_1, H_2 \subset \ol{K}$ of sizes $\delta^r r p_{\max}$ resp. $d + 1$.
	%\item Let $H_2 \subset \ol{K}$ be an arbitrary subset of size $d + 1$.
\end{enumerate}
Denote $\mathrm\Phi_{I,D,p,c}^{(i)} := \mathrm\Phi_{I,D,p,c}(x_i) \in \ol{K}[\term{z}]$
for $i = 1, \dotsc, n$ and define the subset
\[
	\mathcal{H}_{d,r,\delta} = \Bigl\{ \bigl(\mathrm\Phi_{I,D,p,c}^{(1)}(\boldsymbol{a}), \dotsc, \mathrm\Phi_{I,D,p,c}^{(n)}(\boldsymbol{a})\bigr) 
	\,\bigl\vert\; \text{$I \in \tbinom{[n]}{r}$, $p \in [p_{\max}]$, $c \in H_1$, $\boldsymbol{a} \in H_2^r$} \Bigr\} \subset \ol{K}^n .
\]
The following theorem shows that this is a hitting set for the class of circuits under consideration.
A proof is given in Appendix \ref{app:CircuitsWithSparseSubcircuitsArbitraryChar}.

\begin{theorem} \label{thm:CircuitsWithSparseSubcircuitsHittingSetArbitraryChar}
	The set $\mathcal{H}_{d,r,\delta}$ is a hitting set for the class of degree-$d$ circuits with inputs being
	degree-$\delta$ subcircuits of transcendence degree at most $r$. 
	It can be constructed in $\poly(dr\delta n)^{r\delta^{r+1}}$ time.
\end{theorem}

\section{Depth-4 circuits with bounded top and bottom fanin} \label{sec:Depth4Circuits}

The second PIT application of faithful homomorphisms is for
$\spsp_{\delta}(k,s,n)$ circuits. Our hitting set construction is efficient when
the top fanin $k$ and the bottom fanin $\delta$ are both bounded.
Except for top fanin $2$, our hitting set will be {\em conditional} in the
sense that its efficiency depends on a good rank upper bound for depth-$4$ identities.

\subsection{Gcd, simple parts and the rank bounds} \label{sec:RankNotion}

Let $C = \sum_{i=1}^k\prod_{j=1}^s f_{i,j}$ be a  $\spsp_{\delta}(k,s,n)$ circuit, as defined in Sect. \ref{sec:mainResults}. 
Note that the parameters bound the circuit degree, $\deg(C) \le \delta s$. We define an $\Sp(\cdot)$ operator as:
\[
	\Sp(C) := \bigl\{ f_{i,j} \,\vert\; \text{$i \in [k]$ and $j \in [s]$}  \bigr\} \subset K[\term{x}] .
\]
It gives the set of {\em sparse polynomials} of $C$ (wlog we assume them all to be nonzero). The following definitions are 
natural generalizations
of the corresponding concepts for depth-$3$ circuits. Recall $T_i := \prod_j f_{i,j}$, for $i\in[k]$, are the multiplication terms of $C$.
The {\em gcd part} of $C$ is defined as $\gcd(C) := \gcd(T_1, \dotsc, T_k)$ (we fix a unique representative
among the associated $\gcd$s). 
The {\em simple part} of $C$ is defined as $\simple(C) := C/\gcd(C) \in \spsp_{\delta}(k,s,n)$. 
For a subset $I \subseteq [k]$ we denote $C_I := \sum_{i \in I} T_i$. 

Recall that if $C$ is simple then $\gcd(C) = 1$ and if it is minimal then $C_I \neq 0$ for all non-empty $I \subsetneq [k]$. Also, 
recall that $\rk(C)$ is $\trdeg_K\Sp(C)$, and that $R_{\delta}(k, s)$ strictly upper bounds the rank of any minimal and simple 
$\spsp_{\delta}(k,s,n)$ identity. Clearly, $R_{\delta}(k, s)$  is at most $|\Sp(C)|\le ks$ (note: $\Sp(C)$ cannot all be independent in
an identity). On the other hand, we could prove a lower bound on  $R_{\delta}(k, s)$ by constructing identities.  

	From the simple and minimal $\sps$ identities constructed in \cite{bib:SS10}, we obtain the lower bound
	$R_1(k,s) = \mathrm\Omega(k)$ if $\ch(K) = 0$, and $R_1(k,s) = \mathrm\Omega(k \log_p s)$
	if $\ch(K) = p > 0$. These identities can be lifted to $\spsp_{\delta}(k,s,n)$ identities
	by replacing each variable $x_i$ by a product $x_{i,1} \dotsb x_{i,\delta}$ of new variables.
	These examples demonstrate: $R_{\delta}(k,s) = \mathrm\Omega(\delta k)$ if $\ch(K) = 0$, and 
	$R_{\delta}(k,s) = \mathrm\Omega(\delta k \log_p s)$ if $\ch(K) = p > 0$. 
This leads us to the following natural conjecture.

\begin{conj} \label{conj:RankBound}
	We conjecture 
	\[ 
		R_{\delta}(k, s) = \begin{cases}
			\poly(\delta k),        & \text{if $\ch(K) = 0$}, \\
			\poly(\delta k \log s), & \text{otherwise}.
		\end{cases}
	\]
\end{conj}

The following lemma is a vast generalization of \cite[Theorem 3.4]{KSh08} to depth-$4$ circuits.
It suggests how a bound for $R_{\delta}(k,s)$ can be used to construct a hitting set
for $\spsp_{\delta}(k,s,n)$ circuits. %A short proof is given in Appendix \ref{app:RankNotion}.
The $\varphi$ in the statement below should be thought of as a linear map that reduces the number of variables 
from $n$ to $R_{\delta}(k,s)+1$.

\begin{lemma}[Rank is useful] \label{lem:RankBasedBlackBoxAlgo}
	Let $C$ be a $\spsp_{\delta}(k,s,n)$ circuit, let $r := R_{\delta}(k, s)$
	and let $\varphi:K[\term{x}] \rightarrow K[\term{z}] = K[z_0, z_1, \dotsc, z_r]$
	be a linear $K$-algebra homomorphism that, for all $I \subseteq [k]$, satisfies:
	\begin{enumerate}
		\item\label{lem:RankBasedBlackBoxAlgoA} $\varphi(\simple(C_I)) = \simple(\varphi(C_I))$, and
		\item\label{lem:RankBasedBlackBoxAlgoB} $\rk(\varphi(\simple(C_I))) \ge \min\bigl\{ \rk(\simple(C_I)), R_{\delta}(k, s) \bigr\}$.
	\end{enumerate}
	Then $C = 0$ if and only if $\varphi(C) = 0$.
\end{lemma}
\begin{proof}
	If $C = 0$, then clearly $\varphi(C) = 0$. Conversely, let $\varphi(C) = 0$. Let $I \subseteq [k]$
	be a non-empty subset such that $\varphi(C_I)$ is a minimal circuit computing the zero polynomial.
	Then, by assumption (\ref{lem:RankBasedBlackBoxAlgoA}.), $\varphi(\simple(C_I)) = \simple(\varphi(C_I)) \in \spsp_{\delta}(k,s,n)$
	is a minimal and simple circuit computing the zero polynomial. Hence, 
	$\rk(\varphi(\simple(C_I))) < R_{\delta}(k, s)$.
	By assumption (\ref{lem:RankBasedBlackBoxAlgoB}.), this implies $\rk(\varphi(\simple(C_I)))$ $=$
	$\rk(\simple(C_I))$, thus $\varphi$
	is faithful to $\Sp(\simple(C_I))$. Theorem \ref{thm:FaithfulHomomorphismInjective} 
	yields $\simple(C_I) = 0$, hence $C_I = 0$.
	Since $\varphi(C)$ is the sum of zero and minimal circuits $\varphi(C_I)$ for some $I \subseteq [k]$, we obtain
	$C=0$ as required. 
\end{proof}

\subsection{Preserving the simple part (towards Theorem \ref{thm:main2})} \label{sec:PreservingSimplePart}

The following lemma shows that $\mathrm\Psi$ meets condition (\ref{lem:RankBasedBlackBoxAlgoA}.)
of Lemma \ref{lem:RankBasedBlackBoxAlgo}. The proof is given in Appendix \ref{app:PreservingSimplePart}.
This is also the heart of PIT when $k=2$. The actual hitting set, though, we provide in the next subsection.

\begin{lemma}[$\mathrm\Psi$ preserves the simple part] \label{lem:PreservingSimplePart}
	Let $C$ be a $\spsp_{\delta}(k,s,n)$ circuit. Let $D_1 \ge 2\delta^2+1$, let
	$D_1 \ge D_2 \ge \delta+1$ and let $D = (D_1, D_2)$. 
	Then there exists a prime $p\le (2ksn\delta^2)^{8\delta^2+2} (\log_2 D_1)^2+1$ such that any subset 
	$S \subset \ol{K}$ of size $2 \delta^4 k^2s^2p$ contains $c$ satisfying 
	$\mathrm\Psi_{D,p,c}(\simple(C)) = \simple(\mathrm\Psi_{D,p,c}(C))$.
\end{lemma}
\begin{proof}[Proof sketch]
For any coprime $f_i, f_j\in \Sp(C)$ we look at their images under $\mathrm\Psi$. We view $\mathrm\Psi(f_i)$ and $\mathrm\Psi(f_j)$ 
as univariates wrt $z_0$ and fix $z_1=\cdots=z_r=0$. If we could keep these two univariates monic (before the fixing) and their 
resultants nonzero (after the fixing),  
then the coprimality of $\mathrm\Psi(f_i)$ and $\mathrm\Psi(f_j)$  would be ensured. Both those requirements are fulfilled by 
estimating the sparsity and using sparse PIT tricks. 
\end{proof}

\subsection{A hitting set (proving Theorems \ref{thm:main2} \& \ref{thm:main3})} \label{sec:Depth4CircuitsHittingSet}

Armed with Lemmas \ref{lem:RankBasedBlackBoxAlgo} and \ref{lem:PreservingSimplePart} we could now complete the
construction of the hitting set for $\spsp_\delta(k,s,n)$ circuits using the faithful homomorphism $\mathrm\Psi$ with the
right parameters. 

Let $n,\delta,k,s \ge 1$ and let $r = R_{\delta}(k, s)$. We introduce the following parameters.
They are blown up so that they support $2^k$ applications (one for each $I\subset[k]$) of Lemmas 
\ref{lem:VandermondeHomomorphismFaithful} and \ref{lem:PreservingSimplePart}.
\begin{enumerate}
	\item Define $D = (D_1, D_2)$ by $D_1 := (2\delta n)^{2r}$ and $D_2 := \delta+1$.
	\item Define $p_{\max} := 2^{2(k+1)}\cdot(2krsn\delta^2)^{8\delta^2+4\delta r}\lceil\log_2 D_1\rceil^2+1$.
	\item Pick arbitrary $H_1, H_2 \subset \ol{K}$ of sizes $2^{k+2}k^2rs^2\delta^4p_{\max}$ resp. $\delta s + 1$.
	%\item Let $H_2 \subset \ol{K}$ be an arbitrary subset of size $\delta s + 1$.
\end{enumerate}
Denote $\mathrm\Psi_{D,p,c}^{(i)} := \mathrm\Psi_{D,p,c}(x_i) \in \ol{K}[\term{z}]$ for $i = 1, \dotsc, n$
and define the subset
\[
	\mathcal{H}_{\delta,k,s} = \Bigl\{ \bigl(\mathrm\Psi_{D,p,c}^{(1)}(\boldsymbol{a}), \dotsc, \mathrm\Psi_{D,p,c}^{(n)}(\boldsymbol{a})\bigr) 
	\,\bigl\vert\; \text{$p \in [p_{\max}]$, $c \in H_1$, $\boldsymbol{a} \in H_2^{r+1}$} \Bigr\} \subset \ol{K}^n .
\]
The following theorem shows that, in large or zero characteristic, this is a hitting set for $\spsp_{\delta}(k,s,n)$ circuits.

\begin{theorem} \label{thm:Depth4CircuitsHittingSet}
	Assume that $\ch(K) = 0$ or $\ch(K) > \delta^r$. 
	Then $\mathcal{H}_{\delta,k,s}$	is a hitting set for $\spsp_{\delta}(k,s,n)$ circuits.
	It can be constructed in $\poly(\delta rsn)^{\delta^2 kr}$ time.
\end{theorem}

Since trivially $R_{\delta}(2,s) = 1$, we obtain an explicit hitting set for the top fanin $2$ case.
Moreover, in this case we can also eliminate the dependence on the characteristic (because Lemma \ref{lem:PreservingSimplePart} is field independent).

\begin{corollary} \label{cor:Depth4CircuitsHittingSetTopFanIn2}
	Let $K$ be of arbitrary characteristic. Then $\mathcal{H}_{\delta, 2, s}$ is a hitting set for 
	$\spsp_{\delta}(2,s,n)$ circuits. It can be constructed in $\poly(\delta sn)^{\delta^2}$ time.
\end{corollary}

A proof of the theorem and the corollary can be found in Appendix \ref{app:Depth4CircuitsHittingSet}.

\section{Conclusion}

The notion of rank has been quite useful in depth-$3$ PIT. In this work we give the first generalization
of it to depth-$4$ circuits. We used $\trdeg$ and developed fundamental maps -- the faithful homomorphisms --
that preserve $\trdeg$ of sparse polynomials in a blackbox and efficient way (assuming a small $\trdeg$).
Crucially, we showed that faithful homomorphisms preserve the nonzeroness of circuits. 

Our work raises several open questions. The faithful homomorphism construction over a small characteristic
has restricted efficiency, in particular, it is interesting only when the sparse polynomials have very 
low degree. 
Could Lemma \ref{lem:ProjectionHomomorphismFaithful} be improved to handle larger $\delta$?
In general, the classical methods stop short of dealing with small characteristic because the ``geometric''
Jacobian criterion is not
there. We have given some new tools to tackle that, for eg., Corollary \ref{cor:AnnihilatingPolynomialDegreeBound}
and Lemmas \ref{lem:FaithfulHomomorphismExistence} and \ref{lem:ProjectionHomomorphismFaithful}. 
But more tools are needed, for eg. a homomorphism like that
of Lemma \ref{lem:VandermondeHomomorphismFaithful} for arbitrary fields.

Currently, we do not know a better upper bound for $R_\delta(k,s)$ other than $ks$. For $\delta=1$, it is
just the rank of depth-$3$ identities, which is known to be $O(k^2\log s)$ ($O(k^2)$ over $\R$) \cite{SS10focs}.
Even for $\delta=2$ we leave the rank question open. We conjecture $R_2(k,s)=O_k(\log s)$ (generally, 
Conjecture \ref{conj:RankBound}). Our hope is that understanding these small $\delta$ identities should give
us more potent tools to attack depth-$4$ PIT in generality.

\paragraph{Acknowledgements} We are grateful to the Hausdorff Center for Mathematics, Bonn, for its kind support. 
The first two authors would also like to thank the Bonn International Graduate School in Mathematics for research funding.

%+++++++++++++++++++++++++++++++++++++++++++++++++++++++++++++++++++++++++++++++++++++++++++++++++++++
%Bibliography
%+++++++++++++++++++++++++++++++++++++++++++++++++++++++++++++++++++++++++++++++++++++++++++++++++++++

\bibliographystyle{alpha}
\bibliography{pit-refs}

%+++++++++++++++++++++++++++++++++++++++++++++++++++++++++++++++++++++++++++++++++++++++++++++++++++++
%Appendix
%+++++++++++++++++++++++++++++++++++++++++++++++++++++++++++++++++++++++++++++++++++++++++++++++++++++

\begin{appendix}

%\section{Definitions from Commutative Algebra}

\section{Proofs for Sect. \ref{sec:Preliminaries}: Preliminaries} \label{app:Preliminaries}

\subsection{Proofs for Sect. \ref{sec:AlgebraicIndependence}: Perron's criterion} \label{app:AlgebraicIndependence}

For the proof of Corollary \ref{cor:AnnihilatingPolynomialDegreeBound} we will need three well-known
lemmas. The first one is about resultants. For more information about resultants, see \cite{bib:CLO97}.

\begin{lemma}[Resultant] \label{lem:Resultant}
	Let $f, g \in K[\term{x}]$  such that $\deg_{x_i}(f) > 0$ and $\deg_{x_i}(g) > 0$ 
	for some $i \in [n]$. Then $\res_{x_i}(f, g) = 0$ if and only if $f$ and $g$ have a common factor 
	$h \in K[\term{x}]$ with $\deg_{x_i}(h) > 0$.
\end{lemma}
\begin{proof}	
	Cf. \cite[Chap. 3, \S 6, Proposition 1]{bib:CLO97}. 
\end{proof}

The following lemma identifies a situation where annihilating polynomials are unique up to a factor in $K^*$.

\begin{lemma}[Unique annihilating polynomials]\label{lem:RelationIdealPrincipal}
	Let $f_1, \dotsc, f_m \in K[\term{x}]$ contain precisely $m-1$ algebraically independent polynomials
	and let $I \subseteq K[y_1,\dotsc, y_m]$ be the ideal of algebraic relations among $f_1, \dotsc, f_m$. 
	Then $I$ is principal.
\end{lemma}
\begin{proof}
	We follow the instructions of \cite[Exercise 3.2.7]{bib:vdE00}.
	Assume that $f_1, \dotsc, f_{m-1}$ are algebraically independent and let $F_1, F_2 \in K[y_1, \dotsc, y_m]$
	be non-zero irreducible polynomials satisfying $F_i(f_1, \dotsc, f_m) = 0$ for $i=1,2$. It suffices to
	show that $F_1 = c F_2$ for some $c \in K^*$. 
	
	For this, view $F_1, F_2$ as elements of $R[y_m]$, where $R = K[y_1, \dotsc, y_{m-1}]$, and consider the $y_m$-resultant 
	$g := \res_{y_m}(F_1, F_2) \in R$. By \cite[Chap. 3, \S 5, Proposition 9]{bib:CLO97}, there exist 
	$g_1, g_2 \in R[y_m]$ such that $g = g_1 F_1 + g_2 F_2$. We have
	\begin{align*}
		g(f_1, \dotsc, f_{m-1})
		&= g_1(f_1, \dotsc, f_m) \cdot F_1(f_1, \dotsc, f_m) + g_2(f_1, \dotsc, f_m) \cdot F_2(f_1, \dotsc, f_m) \\
		&= 0.
	\end{align*}
	Since $f_1, \dotsc, f_{m-1}$ are algebraically independent, it follows that $g=0$. By Lemma \ref{lem:Resultant},
	$F_1, F_2$ have a non-trivial common factor in $R[y_m]$. Since $F_1, F_2$ are irreducible,
	we obtain $F_1 = c F_2$ for some $c \in K^*$, as required. 
\end{proof}

The following lemma contains a useful fact about annihilating polynomials and algebraic field extensions
(cf. \cite[Claim 7.2]{bib:Kay09} for a similar statement).

\begin{lemma}[Going to a field extension] \label{lem:AnnihilatingPolynomialFieldExtension}
	Let $f_1, \dotsc, f_m \in K[\term{x}]$ and let $L/K$ be an algebraic field extension. If there exists a
	non-zero polynomial $F \in L[\term{y}] = L[y_1, \dotsc, y_m]$ such that $F(f_1, \dotsc, f_m)=0$, then
	there exists a non-zero polynomial $G \in K[\term{y}]$ such that
	$G(f_1, \dotsc, f_m)=0$ and $\deg(G) \le \deg(F)$. In particular, $f_1, \dotsc, f_m$ are algebraically
	independent over $K$ if and only if they are algebraically independent over $L$.
\end{lemma}
\begin{proof}
	Let $F \in L[\term{y}]$ be a non-zero polynomial such that $F(f_1, \dotsc, f_m)=0$.
	Denote by $c_1, \dotsc, c_{\ell} \in L$ the non-zero coefficients of $F$. Replacing
	$L$ by $K(c_1, \dotsc, c_{\ell})$, we may assume that $L/K$ is algebraic and finitely
	generated (as a field) over $K$. By \cite[Chapter V, \S 1, Proposition 1.6]{bib:Lan02},
	this implies that $[L:K] =: d < \infty$. Let $b_1, \dotsc, b_d \in L$ be a $K$-basis of $L$.
	Then we can write $F$ as
	\[
		F = F_1 \cdot b_1 + \dotsb + F_d \cdot b_d
	\]
	for some $F_1, \dotsc, F_d \in K[\term{y}]$, not all zero, such that $\deg(F_i) \le \deg(F)$
	for all $i = 1, \dotsc, d$. Substituting $f_1, \dotsc, f_m$, we obtain
	\[
		0 = F(f_1, \dotsc, f_m) = F_1(f_1, \dotsc, f_m) \cdot b_1 + \dotsb + F_d(f_1, \dotsc, f_m) \cdot b_d .
	\]
	The $K$-linear independence of $b_1, \dotsc, b_d$ implies that all coefficients of
	\[
		F_i(f_1, \dotsc, f_m) \in K[\term{x}]
	\]
	are zero for $i = 1, \dotsc, d$.
	(Here we use that the indeterminates $x_1, \dotsc, x_n$ are $L$-linearly independent,
	because $L/K$ is algebraic.)
	Therefore, some non-zero $F_i$ yields a $G \in K[\term{y}]$ with the desired properties. 
\end{proof}

\noindent {\bf Corollary \ref{cor:AnnihilatingPolynomialDegreeBound}.}
	Let $f_1, \dotsc, f_m \in K[\term{x}]$ be algebraically dependent polynomials of 
	maximal degree $\delta$ and $\trdeg$ $r$.
	Then there exists a non-zero polynomial $F \in K[y_1, \dotsc, y_m]$ of degree at most $\delta^r$ such that
	$F(f_1, \dotsc, f_m) = 0$.
\begin{proof}[Proof of Corollary \ref{cor:AnnihilatingPolynomialDegreeBound}]
	By Lemma \ref{lem:AnnihilatingPolynomialFieldExtension}, we may assume wlog that $K$ is infinite.
	Furthermore, we may assume that $m = r+1$ and $f_1, \dotsc, f_r$ are algebraically independent.
	Let $F \in K[\term{y}] = K[y_1, \dotsc, y_{r+1}]$ be a non-zero {\em irreducible} polynomial such that
	$F(f_1, \dotsc, f_{r+1}) = 0$. By Lemma \ref{lem:FaithfulHomomorphismExistence}, there exists a linear 
	$K$-algebra homomorphism 
	\[
		\varphi: K[\term{x}] \rightarrow K[\term{z}] = K[z_1, \dotsc, z_r]
	\]
	which is faithful to $\{f_1, \dotsc, f_{r+1}\}$. Set $g_i := \varphi(f_i) \in K[\term{z}]$ for $i = 1, \dotsc, r+1$. 
	Then $g_1, \dotsc, g_{r+1}$ are of degree at most $\delta$ and by Theorem \ref{thm:PerronsTheorem} 
	there exists a non-zero polynomial $G \in K[\term{y}]$ such that $G(g_1, \dotsc, g_{r+1}) = 0$
	and $\deg(G) \le \delta^r$. But since
	\[
		F(g_1, \dotsc, g_{r+1}) = F(\varphi(f_1), \dotsc, \varphi(f_{r+1})) = \varphi(F(f_1, \dotsc, f_{r+1})) = 0 ,
	\]
	Lemma \ref{lem:RelationIdealPrincipal} implies that $F$ divides $G$. Hence, $\deg(F) \le \deg(G) \le \delta^r$. 
\end{proof}

\subsection{Proofs for Sect. \ref{sec:JacobianCriterion}: Jacobi's criterion} \label{app:JacobianCriterion}

In the proof of the Jacobian criterion we will make use of the following facts about
partial derivatives. Let $f \in K[\term{x}]$. First assume that $\ch(K) = 0$. Then,
for $i \in [n]$, we have
\[
	\partial_{x_i} f = 0 \qquad\text{if and only if}\qquad f \in K[x_1, \dotsc, x_{i-1}, x_{i+1}, \dotsc, x_n] .
\]
Therefore, we have $\partial_{x_i}(f) = 0$ for all $i = 1, \dotsc, n$ if and only if $f = 0$.
Now assume $\ch(K) = p > 0$. Then, for $i \in [n]$, we have
\[
	\partial_{x_i} f = 0 \qquad\text{if and only if}\qquad f \in K[x_1, \dotsc, x_{i-1}, x_i^p, x_{i+1}, \dotsc, x_n] .
\]
Hence, $\partial_{x_i} f = 0$ for all $i = 1, \dotsc, n$ if and only if $f \in K[x_1^p, \dotsc, x_n^p]$. 
If, in addition, $K$ is a perfect field (in characteristic $p$ this means that every element of $K$ is
a $p$-th power), then we have $\partial_{x_i} f = 0$ for all $i = 1, \dotsc, n$ if and only if
$f = g^p$ for some $g \in K[\term{x}]$. An example of a perfect field is the algebraic closure $\overline{K}$
of $K$.

Now let $K$ be an arbitrary field, let $f_1, \dotsc, f_m \in K[\term{x}]$ and let $F_1, \dotsc, F_s \in K[\term{y}]$. 
Then, by the {\em chain rule}, we have
\begin{multline*}
	J_{\term{x}}(F_1(f_1, \dotsc, f_m), \dotsc, F_s(f_1, \dotsc, f_m)) \\
	= \bigl(J_{\term{y}}(F_1, \dotsc, F_s)\bigr)(f_1, \dotsc, f_m) \cdot J_{\term{x}}(f_1, \dotsc, f_m) .
\end{multline*}
Now we are prepared to proceed with the proofs. \bigskip

\noindent {\bf Lemma \ref{lem:JacobianCriterion}.}
	Let $f_1, \dotsc, f_m \in K[\term{x}]$. Then $\trdeg_K\{f_1, \dotsc,f_m\}\ge$ 
	$\rk_{L} J_{\term{x}}(f_1, \dotsc$, $f_m)$, where $L = K(\term{x})$.

\begin{proof}[Proof of Lemma \ref{lem:JacobianCriterion}]
	Let $r = \rk_{L} J_{\term{x}}(f_1, \dotsc, f_m)$. We may assume that the first $r$ rows of
	$J(f_1, \dotsc, f_m)$ are $L$-linearly independent.
	Assume, for the sake of contradiction, that $f_1, \dotsc, f_r$ are algebraically dependent.
	Choose a non-zero polynomial $F \in K[\term{y}] = K[y_1, \dotsc, y_r]$ of minimal degree such that
	$F(f_1, \dotsc, f_r) = 0$. Differentiating with respect to $x_1, \dotsc, x_n$ using the chain
	rule yields the vector-matrix equation
	\[
		\begin{pmatrix}
			(\partial_{y_1}F)(f_1, \dotsc, f_r), & \dotsc, & (\partial_{y_r}F)(f_1, \dotsc, f_r)
		\end{pmatrix} \cdot
		\begin{pmatrix}
			\partial_{x_1} f_1 & \cdots & \partial_{x_n} f_1 \\
			\vdots & & \vdots \\
			\partial_{x_1} f_r & \cdots & \partial_{x_n} f_r
		\end{pmatrix}  = 0 .
	\]
	Since this matrix has rank $r$ over $L$, it follows that $(\partial_{y_i}F)(f_1, \dotsc, f_r) = 0$
	for all $i = 1, \dotsc, r$. Since the degree of $F$ was chosen to be minimal, it follows that
	$\partial_{y_i} F = 0$ for all $i = 1, \dotsc, r$. If $\ch(K) = 0$, this implies $F = 0$, a
	contradiction. If $\ch(K) = p > 0$, this implies $F \in K[y_1^p, \dotsc, y_r^p]$. Since $\overline{K}$
	is perfect and $F \neq 0$, there is a non-zero $G \in \overline{K}[\term{y}]$ such that $F = G^p$. 
	From
	\[
		0 = F(f_1, \dotsc, f_r) = G(f_1, \dotsc, f_r)^p
	\]
	wee see that $G(f_1, \dotsc, f_r) = 0$. By Lemma \ref{lem:AnnihilatingPolynomialFieldExtension}, there exists
	a non-zero $G' \in K[\term{y}]$ such that $G'(f_1, \dotsc, f_r) = 0$ and $\deg(G') \le \deg(G) < \deg(F)$.
	This contradicts the choice of $F$. 
	Therefore, $f_1, \dotsc, f_r$ are algebraically independent, hence $\trdeg(\{f_1, \dotsc, f_m\}) \ge r$. 
\end{proof}

\noindent {\bf Theorem \ref{thm:JacobianCriterion}.}
	Let $f_1, \dotsc, f_m \in K[\term{x}]$ be polynomials of degree at most $\delta$ and $\trdeg$ $r$. 
	Assume that $\ch(K) = 0$ or $\ch(K) > \delta^r$. Then $\rk_{L} J_{\term{x}}(f_1, \dotsc, f_m)=
	\trdeg_K\{f_1, \dotsc,f_m\}$,
	where $L = K(\term{x})$.

\begin{proof}[Proof of Theorem \ref{thm:JacobianCriterion}]
	Let $r = \trdeg\{f_1, \dotsc, f_m\}$. By Lemma \ref{lem:JacobianCriterion}, we have
	\[
		r \ge \rk_L J(f_1, \dotsc, f_m) ,
	\]
	so it remains to show the converse inequality.
	
	After renumbering $f_1, \dotsc, f_m$ and $x_1, \dotsc, x_n$, we may assume that the polynomials
	$f_1, \dotsc, f_r$, $x_{r+1}, \dotsc, x_n$ are algebraically independent.
	Consequently, for $i = 1, \dotsc, n$, there exist non-zero polynomials $F_i \in K[y_0, \dotsc, y_n]$
	of minimal degree such that $\deg_{y_0}(F_i) > 0$ and
	\begin{equation} \label{eqn:ProofJacobianCriterion}
		F_i(x_i, f_1, \dotsc, f_r, x_{r+1}, \dotsc, x_n) = 0 .
	\end{equation}
	By Theorem \ref{thm:PerronsTheorem} (with $(n-r+1)$ of the $\delta_i$'s being $1$), 
	we have $\deg(F_i) \le \delta^r$. 
	Hence, by the assumptions on $\ch(K)$, 
	we have $\partial_{y_0} F_i \neq 0$. Since the degree of $F_i$ was chosen to be minimal, we have
	\[
		(\partial_{y_0} F_i)(x_i, f_1, \dotsc, f_r, x_{r+1}, \dotsc, x_n) \neq 0 .
	\]
	Denote
	\[
		G_{i,j} := (\partial_{y_j} F_i)(x_i, f_1, \dotsc, f_r, x_{r+1}, \dotsc, x_n)
	\]
	for $j = 0, \dotsc, n$. Differentiating equation \eqref{eqn:ProofJacobianCriterion} with respect to $x_k$ 
	using the chain rule yields
	\[
		G_{i,0} \cdot \delta_{i,k} + \sum_{j=1}^r G_{i,j} \cdot \partial_{x_k} f_j + \sum_{j=r+1}^n G_{i,j} \cdot \delta_{j,k} = 0
	\]
	for $k = 1, \dotsc, n$. Since $G_{i,0} \neq 0$, this can be rewritten as
	\[
		\sum_{j=1}^r \frac{-G_{i,j}}{G_{i,0}} \cdot \partial_{x_k} f_j + 
		\sum_{j=r+1}^n \frac{-G_{i,j}}{G_{i,0}} \cdot \delta_{j,k} = \delta_{i,k} .
	\]
	This shows that the block diagonal matrix
	\[
		\begin{pmatrix}
			\partial_{x_1} f_1 & \cdots & \partial_{x_r} f_1 & & & \\
			\vdots & & \vdots & & & \\
			\partial_{x_1} f_r & \cdots & \partial_{x_r} f_r & & & \\
			& & & \enspace 1 \enspace & & \\
			& & & & \ddots & \\
			& & & & & \enspace 1 \enspace
		\end{pmatrix} \in L^{n \times n}
	\]
	is invertible. Therefore, the first $r$ rows of $J(f_1, \dotsc, f_m)$ are $L$-linearly independent
	and hence $r \le \rk_L J(f_1, \dotsc, f_m)$. 
\end{proof}			

\subsection{Proofs for Sect. \ref{sec:AffineAlgebrasKrullDimension}: Krull dimension} \label{app:AffineAlgebrasKrullDimension}

\noindent {\bf Corollary \ref{cor:DimNonIncreasing}.}
	Let $A,B$ be $K$-algebras and let $\varphi:A \rightarrow B$ be a 
	$K$-algebra homomorphism. If $A$ is an affine algebra, then so is $\varphi(A)$ and we have 
	$\dim(\varphi(A)) \le \dim(A)$. 
	If, in addition, $\varphi$ is injective, then $\dim(\varphi(A)) = \dim(A)$.
\begin{proof}[Proof of Corollary \ref{cor:DimNonIncreasing}]
	Since $A$ is an affine algebra, there exist $a_1, \dotsc, a_m \in A$ such that $A = K[a_1, \dotsc, a_m]$.
	Then $\varphi(A) = K[\varphi(a_1), \dotsc, \varphi(a_m)]$ is finitely generated as a $K$-algebra as well.
	
	Now assume for the sake of contradiction that $d := \dim(\varphi(A)) > \dim(A)$. By Theorem 
	\ref{thm:DimEqualsTrdeg}, there exist $a_1, \dotsc, a_d \in A$ such that 
	$\varphi(a_1), \dotsc, \varphi(a_d)$ are algebraically independent.
	Since $d > \dim(A)$, the elements $a_1, \dotsc, a_d$ are algebraically dependent. Hence, there
	exists a non-zero polynomial $F \in K[y_1, \dotsc, y_d]$ such that $F(a_1, \dotsc, a_d) = 0$.
	It follows that
	\[
		0 = \varphi(F(a_1, \dotsc, a_d)) = F(\varphi(a_1), \dotsc, \varphi(a_d))
	\]
	and this implies that $\varphi(a_1), \dotsc, \varphi(a_d)$ are algebraically dependent, a contradiction.
	Therefore, $\dim(\varphi(A)) \le \dim(A)$.
	
	Now let $\varphi$ be injective, let $d := \dim(A)$ and let $a_1, \dotsc, a_d \in A$ be algebraically independent.
	Assume for the sake of contradiction that $\varphi(a_1), \dotsc, \varphi(a_d)$ are algebraically dependent.
	Then there exists a non-zero polynomial $F \in K[y_1, \dotsc, y_d]$ such that $F(\varphi(a_1), \dotsc, \varphi(a_d)) = 0$.
	From
	\[
		0 = F(\varphi(a_1), \dotsc, \varphi(a_d)) = \varphi(F(a_1, \dotsc, a_d))
	\]
	we see that $F(a_1, \dotsc, a_d) = 0$, because $\varphi$ is injective. But this means that $a_1, \dotsc, a_d$ are
	algebraically dependent, a contradiction. Thus $\dim(\varphi(A)) \ge \dim(A)$. 
\end{proof}

\section{Proofs for Sect. \ref{sec:FaithfulHomomorphisms}: Faithful homomorphisms} \label{app:FaithfulHomomorphisms}

Let $\PP$ denote the set of prime numbers and $\sparse(f)$ denote the {\em sparsity} of a polynomial $f$.

In the proofs of Lemmas \ref{lem:ProjectionHomomorphismFaithful}, \ref{lem:VandermondeHomomorphismFaithful}
and \ref{lem:PreservingSimplePart} we will use the following well-known facts.

\begin{lemma}[Sparse PIT] \label{lem:SparsePIT}
	Let $\ell \ge 1$ and $d \ge 2$. Let $R$ be a commutative ring and let $f \in R[t]$ be a non-zero polynomial
	of sparsity at most $\ell$ and degree at most $d$. Then there are at most $\ell \cdot \log_2(d)-1$ prime
	numbers $p$ such that $f = 0 \pmod{\langle t^p-1 \rangle_{R[t]}}$.
\end{lemma}
\begin{proof}
	Cf. \cite[Lemma 13]{BHLV09} and note that the given proof also works for polynomials over a ring
	(instead of a field). 
\end{proof}

\begin{lemma}[Primes] \label{lem:NumberOfPrimes}
	Let $r\in\R^{\ge2}$. Then the interval $[1,r^2+1]$ contains at least $\lceil r\rceil$ prime numbers.
\end{lemma}
\begin{proof}
	Cf. \cite[Claim on p. 478]{bib:Pap95}. 
\end{proof}

\subsection{Proofs for Sect. \ref{sec:FaithfulHomomorphismExistence}: A Kronecker-inspired map} \label{app:FaithfulHomomorphismExistence}

\noindent {\bf Lemma \ref{lem:FaithfulHomomorphismExistence}.}
	Let $K$ be an infinite field and let $f_1, \dotsc, f_m \in K[\term{x}]$ be polynomials of 
	$\trdeg$ $r$. Then there exists a linear $K$-algebra homomorphism 
	$\varphi: K[\term{x}] \rightarrow K[\term{z}]$ which is faithful to
	$\{f_1, \dotsc, f_m\}$.
\begin{proof}[Proof of Lemma \ref{lem:FaithfulHomomorphismExistence}]
	After renumbering $f_1, \dotsc, f_m$ and $x_1, \dotsc, x_n$, we may assume that 
	$f_1, \dotsc, f_r$, $x_{r+1}, \dotsc, x_n$ are algebraically independent.
	Consequently, for $i = 1, \dotsc, r$, there exists a non-zero polynomial
	$G_i \in K[y_0, y_1, \dotsc, y_n]$ such that $\deg_{y_0}(G_i) > 0$ and
	\[
		G_i(x_i, f_1, \dotsc, f_r, x_{r+1}, \dotsc, x_n) = 0 .
	\]
	Denote by $g_i \in K[y_1, \dotsc, y_n]$ the (non-zero) leading coefficient of $G_i$
	as a polynomial in $y_0$ with coefficients in $K[y_1, \dotsc, y_n]$.
	The algebraic independence of $f_1, \dotsc, f_r$, $x_{r+1}, \dotsc, x_n$ implies
	\[
		g_i(f_1, \dotsc, f_r, x_{r+1}, \dotsc, x_n) \neq 0 .
	\]
	Since $K$ is infinite, there exist $c_{r+1}, \dotsc, c_n \in K$ such that
	\[
		(g_i(f_1, \dotsc, f_r, x_{r+1}, \dotsc, x_n))(x_1, \dotsc, x_r, c_{r+1}, \dotsc, c_n) \neq 0
	\]
	for all $i = 1, \dotsc, r$. Now define the $K$-algebra homomorphism
	\[
		\varphi: K[\term{x}] \rightarrow K[\term{z}], \qquad x_i \mapsto
		\begin{cases}
			z_i, & \text{if $1 \le i \le r$}, \\
			c_i, & \text{otherwise} .
		\end{cases}
	\]
	Then, by the choice of $c_{r+1}, \dotsc, c_n$, we have
	\[
		G_i(y_0, \varphi(f_1), \dotsc, \varphi(f_r), c_{r+1}, \dotsc, c_n) \neq 0 
	\]
	and
	\[
		G_i(z_i, \varphi(f_1), \dotsc, \varphi(f_r), c_{r+1}, \dotsc, c_n) = 0 
	\]
	for $i = 1, \dotsc, r$.	This shows that $z_i$ is algebraically dependent on 
	$\varphi(f_1), \dotsc, \varphi(f_r)$ for $i = 1, \dotsc, r$. It follows that 
	\[
		\trdeg\{\varphi(f_1), \dotsc, \varphi(f_m)\} = r = \trdeg\{f_1, \dotsc, f_m\} ,
	\]
	hence $\varphi$ is faithful to $\{f_1, \dotsc, f_m\}$. 
\end{proof}

\noindent {\bf Lemma \ref{lem:ProjectionHomomorphismFaithful}.}
	Let $f_1, \dotsc, f_m \in K[\term{x}]$ be polynomials of degree at most $\delta$ and 
	$\trdeg$ at most $r$. Let $D > \delta^{r+1}$.
	Then there exist an index set $I \in \tbinom{[n]}{r}$ and a prime $p\le(n+\delta^r)^{8\delta^{r+1}}(\log_2D)^2+1$
	such that any subset of $\ol{K}$ of size $\delta^r r p$ contains $c$ such that
	$\mathrm\Phi_{I,D,p,c}$ is faithful to $\{f_1, \dotsc, f_m\}$.
\begin{proof}[Proof of Lemma \ref{lem:ProjectionHomomorphismFaithful}]
	We may assume wlog that $\trdeg\{f_1, \dotsc, f_m\} = r$ and,
	after renumbering $f_1, \dotsc, f_m$, that 
	\[
		f_1, \dotsc, f_r, x_{j_{r+1}}, \dotsc, x_{j_n}
	\]
	are algebraically independent for some $j_{r+1}, \dotsc, j_n \in [n]$ with $j_{r+1} < \dotsb < j_n$.
	Denote the complement $[n] \setminus \{j_{r+1}, \dotsc, j_n\}$ by	$I = \{j_1, \dotsc, j_r\}$, where $j_1 < \dotsb < j_r$.
	By Corollary \ref{cor:AnnihilatingPolynomialDegreeBound}, there exists a non-zero polynomial $G_i \in K[y_0, y_1, \dotsc, y_n]$
	such that $\deg(G_i) \le \delta^r$, $\deg_{y_0}(G_i) > 0$ and
	\[
		G_i(x_{j_i}, f_1, \dotsc, f_r, x_{j_{r+1}}, \dotsc, x_{j_n}) = 0
	\]
	for $i = 1, \dotsc, r$.	Denote by $g_i \in K[y_1, \dotsc, y_n]$ the (non-zero) leading coefficient of 
	$G_i$ as a polynomial $y_0$ with coefficients in $K[y_1, \dotsc, y_n]$. The algebraic independence of
	$f_1, \dotsc, f_r$, $x_{j_{r+1}}, \dotsc, x_{j_n}$ implies
	\[
		g_i(f_1, \dotsc, f_r, x_{j_{r+1}}, \dotsc, x_{j_n}) \neq 0 .
	\]
	We have
	\[
		\deg\bigl( g_i(f_1, \dotsc, f_r, x_{j_{r+1}}, \dotsc, x_{j_n}) \bigr) \le \delta^{r+1} < D .
	\]
	Therefore, the polynomial
	\[
		h_i := g_i(\mathrm\Phi_{I,D}(f_1), \dotsc, \mathrm\Phi_{I,D}(f_r), 
		\mathrm\Phi_{I,D}(x_{j_{r+1}}), \dotsc, \mathrm\Phi_{I,D}(x_{j_n})) \in K[t, \term{z}]
	\]
	is non-zero (this is the classical Kronecker substitution: $D$ is so large that the monomials remain
	separated). We have
	\[
		\deg_t(h_i) \le \delta^{r+1} \cdot (D + D^2 + \dotsb + D^{n-r}) \le D^{n+1} .
	\]
	Also, the sparsity of $h_i$ (short, $\sparse$) can be bounded as: 
	\begin{align*}
		 \sparse(h_i) &= \sparse\bigl( g_i(f_1, \dotsc, f_r, x_{j_{r+1}}, \dotsc, x_{j_n}) \bigr) \\
		 &\le \sparse(g_i) \cdot \max\{ \sparse(f_1), \dotsc, \sparse(f_r)\}^{\deg(g_i)} \\
		 &\le \binom{n+\delta^r}{\delta^r} \cdot \binom{n+\delta}{\delta}^{\delta^r} \\
		 &\le (n+\delta^r)^{\delta^r} \cdot (n+\delta)^{\delta^{r+1}} .
	\end{align*}
	Let $B_i \subseteq \PP$ be the set of all primes $p$ satisfying
	$h_i = 0 \pmod{\langle t^p-1 \rangle_{K[t, \term{z}]}}$. Then
	$\abs{B_i} < (n+1)(n+\delta^r)^{\delta^r}(n+\delta)^{\delta^{r+1}}\log_2D$ by Lemma \ref{lem:SparsePIT}.
	Finally set $B := B_1 \cup \dotsb \cup B_r$. Then
	\[
		\abs{B} < r(n+1)(n+\delta^r)^{\delta^r}(n+\delta)^{\delta^{r+1}}\log_2D \le (n+\delta^r)^{4\delta^{r+1}}\log_2D .
	\]
	Now pick a suitable prime $p \in \PP \setminus B$ (by Lemma \ref{lem:NumberOfPrimes}). 
    Let $i \in [r]$. Then $h_i \neq 0 \pmod{\langle t^p-1 \rangle_{K[t, \term{z}]}}$. Define
	\[
		h_i^{(p)} := g_i(\mathrm\Phi_{I,D,p}(f_1), \dotsc, \mathrm\Phi_{I,D,p}(f_r), 
		\mathrm\Phi_{I,D,p}(x_{j_{r+1}}), \dotsc, \mathrm\Phi_{I,D,p}(x_{j_n})) \in K[t, \term{z}] .
	\]
	Since $h_i^{(p)} = h_i \neq 0 \pmod{\langle t^p-1 \rangle_{K[t, \term{z}]}}$, we have $h_i^{(p)} \neq 0$.
	Let $S_i \subset \ol{K}$ be the set of all $c \in \ol{K}$ such that $h_i^{(p)}(c, \term{z}) = 0$.
	Then $\abs{S_i} \le \deg_t(h_i^{(p)}) < \delta^r p$. Finally set $S := S_1 \cup \dotsb \cup S_r$. Then
	$\abs{S} < r \delta^r p$. 
	
	Now let $i \in [r]$ and $c \in \ol{K} \setminus S$. Then
	\[
		G_i\bigl(y_0, \mathrm\Phi_{I,D,p,c}(f_1), \dotsc, \mathrm\Phi_{I,D,p,c}(f_r), 
		c^{\lfloor D^1 \rfloor_p}, \dotsc, c^{\lfloor D^{n-r} \rfloor_p} \bigr) \neq 0 ,
	\]
	because $h_i^{(p)}(c, \term{z}) \neq 0$, and
	\[
		G_i\bigl(z_i, \mathrm\Phi_{I,D,p,c}(f_1), \dotsc, \mathrm\Phi_{I,D,p,c}(f_r), 
		c^{\lfloor D^1 \rfloor_p}, \dotsc, c^{\lfloor D^{n-r} \rfloor_p} \bigr) = 0 .
	\]
	This shows that	$z_i$ is algebraically dependent on $\mathrm\Phi_{I,D,p,c}(f_1), \dotsc, \mathrm\Phi_{I,D,p,c}(f_r)$
	for $i = 1, \dotsc, r$. It follows that
	\[
		\trdeg\{\mathrm\Phi_{I,D,p,c}(f_1), \dotsc, \mathrm\Phi_{I,D,p,c}(f_m)\} = r = \trdeg\{f_1, \dotsc, f_m\}
	\]
	for all $c \in \ol{K} \setminus S$. 
\end{proof}

\subsection{Proofs for Section \ref{sec:VandermondeHomomorphism}: A Vandermonde-inspired map} \label{app:VandermondeHomomorphism}

\noindent {\bf Lemma \ref{lem:VandermondeHomomorphismFaithful}.}
	Let $f_1, \dotsc, f_m \in K[\term{x}]$ be polynomials of sparsity at most $\ell$, degree at most $\delta$
	and $\trdeg$ at most $r$. Assume that $\ch(K)=0$ or $\ch(K) > \delta^r$. Let $D = (D_1, D_2)$ such that
	$D_1 \ge \max\{\delta r+1, (n+1)^{r+1}\}$ and $D_2 \ge 2$. 	
	Then there exists a prime $p\le (2nr\ell)^{2(r+1)} (\log_2 D_1)^2+1$ such that any subset of $\ol{K}$ of size 
  $\delta r p$ contains $c$ such that
	$\mathrm\Psi_{D,p,c}$ is faithful to $\{f_1, \dotsc, f_m\}$.
\begin{proof}[Proof of Lemma \ref{lem:VandermondeHomomorphismFaithful}]
	Let $s := \trdeg\{f_1, \dotsc, f_m\} \le r$ and let $i_1, \dotsc, i_s \in [m]$ such that
	$f_{i_1}, \dotsc, f_{i_s}$ are algebraically independent. By the chain rule, we have
	\begin{multline} \label{eqn:ProofVandermondeHomomorphismFaithful1}
		J_{z_1, \dotsc, z_s}(\mathrm\Psi_D(f_{i_1}), \dotsc, \mathrm\Psi_D(f_{i_s})) \\
		= \bigl( J_{\term{x}}(f_{i_1}, \dotsc, f_{i_s}) \bigr)(\mathrm\Psi_D(x_1), \dotsc, \mathrm\Psi_D(x_n)) 
		\cdot J_{z_1, \dotsc, z_s}(\mathrm\Psi_D(x_1), \dotsc, \mathrm\Psi_D(x_n)).
	\end{multline}
	We introduce some notation. Define the polynomial
	\[
		f' := \det J_{z_1, \dotsc, z_s}(\mathrm\Psi_D(f_{i_1}), \dotsc, \mathrm\Psi_D(f_{i_s})) \in K[t, \term{z}]
	\]
	and set $f := f'(t, 0, \dotsc, 0) \in K[t]$. For an index set $I = \{j_1, \dotsc, j_s\} \in \tbinom{[n]}{s}$ with 
	$j_1 < \dotsb < j_s$, denote
	\[
		g'_I := \bigl( \det J_{x_{j_1}, \dotsc, x_{j_s}}(f_{i_1}, \dotsc, f_{i_s}) \bigr)(\mathrm\Psi_D(x_1), \dotsc, \mathrm\Psi_D(x_n)) 
		\in K[t, \term{z}]
	\]
	and
	\[
		h'_I := \det J_{z_1, \dotsc, z_s}(\mathrm\Psi_D(x_{j_1}), \dotsc, \mathrm\Psi_D(x_{j_s})) \in K[t, \term{z}] ,
	\]
	and set $g_I := g'_I(t, 0 \dotsc, 0) \in K[t]$ and $h_I := h'_I(t, 0, \dotsc, 0) \in K[t]$.
	Applying the Cauchy-Binet formula (cf. \cite{bib:Zen93}) to \eqref{eqn:ProofVandermondeHomomorphismFaithful1}
	and substituting $(t, 0, \dotsc, 0)$ for $(t, z_0, \dotsc, z_r)$, we obtain
	\begin{equation} \label{eqn:ProofVandermondeHomomorphismFaithful2}
		f = \sum_{I \in \mathcal{I}} g_I \cdot h_I ,
	\end{equation}
	where $\mathcal{I} := \{ I \in \tbinom{[n]}{s} \,\vert\; g_I \neq 0 \}$. We want to prove that $f \neq 0$.
	It suffices to show that there is a unique $I \in \mathcal{I}$ for which $\deg(g_I \cdot h_I)$ is maximal.
	
	First we show that $\mathcal{I} \neq \varnothing$. Since $f_{i_1}, \dotsc, f_{i_s}$ are algebraically independent,
	there exists $I = \{j_1, \dotsc, j_s\} \in \tbinom{[n]}{s}$ with $j_1 < \dotsb < j_s$ such that
	\[
		\det J_{x_{j_1}, \dotsc, x_{j_s}}(f_{i_1}, \dotsc, f_{i_s}) \neq 0
	\]
	by Theorem \ref{thm:JacobianCriterion}. We have
	\[
		\deg\bigl( \det J_{x_{j_1}, \dotsc, x_{j_s}}(f_{i_1}, \dotsc, f_{i_s}) \bigr) \le \delta s \le \delta r .
	\]
	Since $D \ge \delta r+1$, it follows that $g_I \neq 0$ (this is the classical Kronecker substitution: 
	$D$ is so large that the monomials remain separated), hence $I \in \mathcal{I}$.
	
	Next we want to show that $h_I \neq 0$ and $\deg(h_I) < D$ for all $I \in \tbinom{[n]}{s}$,
	and we want to show that $\deg(h_I) \neq \deg(h_{I'})$ for all $I, I' \in \tbinom{[n]}{s}$ with $I \neq I'$.
	To this end, let $I = \{j_1, \dotsc, j_s\} \in \tbinom{[n]}{s}$ with $j_1 < \dotsb < j_s$. Then
	\[
		h_I = \det\begin{pmatrix}
			t^{j_1(n+1)^1} & \cdots & t^{j_1(n+1)^s} \\
			\vdots & & \vdots \\
			t^{j_s(n+1)^1} & \cdots & t^{j_s(n+1)^s}
		\end{pmatrix}
		= \sum_{\sigma \in \mathfrak{S}_s} \sgn(\sigma) \cdot t^{d_{\sigma}} ,
	\]
	where $\mathfrak{S}_s$ denotes the symmetric group on $\{1, \dotsc, s\}$
	and 
	\[
		d_{\sigma} := j_1 (n+1)^{\sigma(1)} + \dotsb + j_s (n+1)^{\sigma(s)} \in \N .
	\]
	It is not hard to show that $d_{\id} > d_{\sigma}$ for all 
	$\sigma \in \mathfrak{S}_s \setminus \{\id\}$. This implies $h_I \neq 0$
	and
	\[
		\deg(h_I) = j_1 (n+1)^1 + \dotsb + j_s (n+1)^s < (n+1)^{s+1} \le (n+1)^{r+1} \le D .
	\]
	From the degree formula it is not hard to deduce that $\deg(h_I) \neq \deg(h_{I'})$ for
	all $I, I' \in \tbinom{[n]}{s}$ with $I \neq I'$.
	
	Now denote by $\mathcal{I}_{\max} \subseteq \mathcal{I}$ the set of all $I \in \mathcal{I}$ such that
	$\deg(g_I)$ is maximal. Let $I \in \mathcal{I}_{\max}$ and let $I' \in \mathcal{I}\setminus\mathcal{I}_{\max}$.
	Observe that, by construction, we have $\deg(g_I)-\deg(g_{I'}) \ge D$.
	Since $\deg(h_{I'}) < D$, it follows that
	\[
		\deg(g_I \cdot h_I) \ge \deg(g_I) \ge \deg(g_{I'}) + D > \deg(g_{I'}) + \deg(h_{I'}) = \deg(g_{I'} \cdot h_{I'}) .
	\]
	Therefore, the summands in \eqref{eqn:ProofVandermondeHomomorphismFaithful2} of maximal degree have
	an index set in $\mathcal{I}_{\max}$.
	
	Finally, let $I \in \mathcal{I}_{\max}$ be the unique index set such that $\deg(h_I)$ is maximal. Then
	$g_I \cdot h_I$ is the unique summand in \eqref{eqn:ProofVandermondeHomomorphismFaithful2}
	of maximal degree. This implies $f \neq 0$, as required.
	
	By \eqref{eqn:ProofVandermondeHomomorphismFaithful2}, we have
	\[
		\sparse(f) \le \binom{n}{s} \cdot (s! \cdot \ell^s) \cdot s! \le (ns \ell)^s \le (nr\ell)^r
	\]
	and
	\[
		\deg(f) \le r\delta \cdot (D_1 + D_1^2 + \dotsb + D_1^n) + (n+1)^{r+1} \le D_1^{n+1} + D_1 \le D_1^{n+2} .
	\]
	Let $B \subseteq \PP$ be the set of all primes $p$ satisfying
	$f = 0 \pmod{\langle t^p-1 \rangle_{K[t]}}$. Then
	\[
		\abs{B} < (n+2) (nr\ell)^r \log_2 D_1 \le (2nr\ell)^{r+1} \log_2 D_1
	\]
	by Lemma \ref{lem:SparsePIT}.
	
	Now pick a suitable prime $p \in \PP \setminus B$ (by Lemma \ref{lem:NumberOfPrimes}). Then $f \neq 0 \pmod{\langle t^p-1 \rangle_{K[t]}}$.
	This implies $f' \neq 0 \pmod{\langle t^p-1 \rangle_{K[t,\term{z}]}}$. Define
	\[
		f^{(p)} := \det J_{z_1, \dotsc, z_s}(\mathrm\Psi_{D,p}(f_{i_1}), \dotsc, \mathrm\Psi_{D,p}(f_{i_s})) \in K[t, \term{z}] .
	\]
	Since $f^{(p)} = f' \neq 0 \pmod{\langle t^p-1 \rangle_{K[t,\term{z}]}}$, we have $f^{(p)} \neq 0$.
	Let $S \subset \ol{K}$ be the set of all $c \in \ol{K}$ such that $f^{(p)}(c, \term{z}) = 0$.
	Then $\abs{S} \le \deg_t(f^{(p)}) < \delta s p \le \delta r p$. Now let $c \in \ol{K} \setminus S$. Then
	\[
		\det J_{z_1, \dotsc, z_s}(\mathrm\Psi_{D,p,c}(f_{i_1}), \dotsc, \mathrm\Psi_{D,p,c}(f_{i_s})) = f^{(p)}(c, \term{z}) \neq 0 .
	\]
	By Theorem \ref{thm:JacobianCriterion}, this means that $\mathrm\Psi_{D,p,c}(f_{i_1}), \dotsc, \mathrm\Psi_{D,p,c}(f_{i_s})$
	are algebraically independent, hence
	\[
		\trdeg\{\mathrm\Psi_{D,p,c}(f_1), \dotsc, \mathrm\Psi_{D,p,c}(f_m)\} = s = \trdeg\{f_1, \dotsc, f_m\}
	\]
	for all $c \in \ol{K} \setminus S$. 
\end{proof}

\section{Proofs for Sect. \ref{sec:CircuitsWithSparseSubcircuits}: Proving Theorem \ref{thm:main1}} \label{app:CircuitsWithSparseSubcircuits}

\subsection{Proofs for Sect. \ref{sec:CircuitsWithSparseSubcircuitsCharZero}: A hitting set} 
\label{app:CircuitsWithSparseSubcircuitsCharZero}

\noindent {\bf Theorem \ref{thm:CircuitsWithSparseSubcircuitsHittingSetCharZero}.}
	Assume that $\ch(K) = 0$ or $\ch(K) > \delta^r$. Then $\mathcal{H}_{d,r,\delta,\ell}$ is
	a hitting set for the class of degree-$d$ circuits with inputs being $\ell$-sparse, degree-$\delta$ subcircuits
	of $\trdeg$ at most $r$. It can be constructed in $\poly(dr\delta\ell n)^r$ time.
\begin{proof}[Proof of Theorem \ref{thm:CircuitsWithSparseSubcircuitsHittingSetCharZero}]
	Let $C(f_1, \dotsc, f_m)$ be a non-zero circuit of degree at most $d$ with subcircuits
	$f_1, \dotsc, f_m$ of sparsity at most $\ell$, degree at most $\delta$ and $\trdeg$ at most $r$. 
	By the choice of parameters, Lemma \ref{lem:VandermondeHomomorphismFaithful} implies that there exist 
	a prime $p \in [p_{\max}]$ 
	and an element $c \in H_1$ such that $\mathrm\Psi_{D,p,c}$ is faithful to $\{f_1, \dotsc, f_m\}$.
	Hence, by Theorem \ref{thm:FaithfulHomomorphismInjective},
	\[
		\mathrm\Psi_{D,p,c}(C(f_1, \dotsc, f_m)) = C(\mathrm\Psi_{D,p,c}(f_1), \dotsc, \mathrm\Psi_{D,p,c}(f_m))
	\]
	is a non-zero circuit with at most $r+1$ variables and of degree at most $d$. Now the first assertion follows from
	Lemma \ref{lem:CombinatorialNullstellensatz}. The second assertion is obvious from the construction. 
\end{proof}

\subsection{Proofs for Sect. \ref{sec:CircuitsWithSparseSubcircuitsArbitraryChar}: Arbitrary characteristic} 
\label{app:CircuitsWithSparseSubcircuitsArbitraryChar}

\noindent {\bf Theorem \ref{thm:CircuitsWithSparseSubcircuitsHittingSetArbitraryChar}.}
	The set $\mathcal{H}_{d,r,\delta}$ is a hitting set for the class of degree-$d$ circuits with inputs being
	degree-$\delta$ subcircuits of transcendence degree at most $r$. 
	It can be constructed in $\poly(dr\delta n)^{r\delta^{r+1}}$ time.
\begin{proof}[Proof of Theorem \ref{thm:CircuitsWithSparseSubcircuitsHittingSetArbitraryChar}]
	Let $C(f_1, \dotsc, f_m)$ be a non-zero circuit of degree at most $d$ with subcircuits
	$f_1, \dotsc, f_m$ of degree at most $\delta$ and $\trdeg$ at most $r$. 
	By the choice of parameters, Lemma \ref{lem:ProjectionHomomorphismFaithful} implies that there exist 
	an index set $I \in \tbinom{[n]}{r}$, a prime $p \in [p_{\max}]$ 
	and an element $c \in H_1$ such that $\mathrm\Phi_{I,D,p,c}$ is faithful to $\{f_1, \dotsc, f_m\}$.
	Hence, by Theorem \ref{thm:FaithfulHomomorphismInjective},
	\[
		\mathrm\Phi_{I,D,p,c}(C(f_1, \dotsc, f_m)) = C(\mathrm\Phi_{I,D,p,c}(f_1), \dotsc, \mathrm\Phi_{I,D,p,c}(f_m))
	\]
	is a non-zero circuit with at most $r$ variables and of degree at most $d$. Now the first assertion follows from
	Lemma \ref{lem:CombinatorialNullstellensatz}. The second assertion is obvious from the construction. 
\end{proof}

\section{Proofs for Sect. \ref{sec:Depth4Circuits}: Depth-4 circuits} \label{app:Depth4Circuits}

\subsection{Proofs for Sect. \ref{sec:PreservingSimplePart}: Preserving the simple part} \label{app:PreservingSimplePart}

\noindent {\bf Lemma \ref{lem:PreservingSimplePart}.}
	Let $C$ be a $\spsp_{\delta}(k,s,n)$ circuit. Let $D_1 \ge 2\delta^2+1$, let
	$D_1 \ge D_2 \ge \delta+1$ and let $D = (D_1, D_2)$. 
	Then there exists a prime $p\le (2ksn\delta^2)^{8\delta^2+2} (\log_2 D_1)^2$ $+1$ such that any subset 
	$S \subset \ol{K}$ of size $2\delta^4k^2s^2p$ contains $c$ satisfying 
	$\mathrm\Psi_{D,p,c}(\simple(C))$ $= \simple(\mathrm\Psi_{D,p,c}(C))$.
\begin{proof}[Proof of Lemma \ref{lem:PreservingSimplePart}]
	Let $f_1, \dotsc, f_m \in K[\term{x}]$ be the non-constant {\em irreducible} factors
	of the polynomials in $\Sp(C)$. Then $m \le ks\delta$ and we have
	\[
		\deg(f_i) \le \delta \qquad\text{and}\qquad \sparse(f_i) \le \binom{n+\delta}{\delta} \le (n+\delta)^{\delta}
	\]
	for all $i = 1, \dotsc, m$. 
	
	First we make the following observation. If $\varphi: K[\term{x}] \rightarrow K[\term{z}]$ is a 
	$K$-algebra homomorphism such that 
	\begin{enumerate}
		\item\label{enum:ProofPreservingSimplePartA} $\varphi(f_i)$ is non-constant, 
		for all $i = 1, \dotsc, m$, and
		\item\label{enum:ProofPreservingSimplePartB} $\gcd(f_i, f_j) = 1$ implies 
		$\gcd(\varphi(f_i), \varphi(f_j)) = 1$, for all $1 \le i < j \le m$,
	\end{enumerate}
	then $\varphi(\simple(C)) = \simple(\varphi(C))$. To satisfy the first condition we will ensure that 
	the images of $f_1, \dotsc, f_m$ under $\mathrm\Psi$ are monic in $z_0$. This will also facilitate our task of meeting
	the second condition. Here we will use resultants with respect to $z_0$ to preserve coprimality.
	
	So let $i \in [m]$ and define
	\[
		g_i := f_i\bigl(t^{D_2^1}, \dotsc, t^{D_2^n}\bigr) \in K[t] . 
	\]
	Since $\deg(f_i) < D_2$, we have $g_i \neq 0$ (Kronecker substitution). We have
	\[
		\deg(g_i) \le \delta \cdot (D_2 + D_2^2 + \dotsb + D_2^n) \le D_2^{n+1}
	\]
	and $\sparse(g_i) = \sparse(f_i) \le (n+\delta)^{\delta}$.
	Let $B_{1,i} \subseteq \PP$ be the set of all primes $p$ satisfying $g_i = 0 \pmod{\langle t^p-1 \rangle_{K[t]}}$.
	Then $\abs{B_{1,i}} < (n+1)(n+\delta)^{\delta} \log_2D_2$ by Lemma \ref{lem:SparsePIT}. 
	Finally, set $B_1 := B_{1,1} \cup \dotsb \cup B_{1,m}$.
	Then 
	\[
		\abs{B_1} \le m(n+1)(n+\delta)^{\delta} \log_2D_2 \le ks\delta(n+1)(n+\delta)^{\delta} \log_2D_2 .
	\]
	
	Now let $i \in [m]$ and define
	\[
		h_i := f_i\bigl(x_1 + t^{D_2^1}z_0, \dotsc, x_n + t^{D_2^n}z_0 \bigr) \in K[t, z_0, \term{x}] .
	\]
	Then the leading term of $h_i$ as a polynomial in $z_0$ is $g_i$. In particular, $h_i \neq 0$.
	We have
	\[
		\sparse(h_i) \le 2^{\delta} \cdot \sparse(f_i) \le 2^{\delta}(n+\delta)^{\delta} .
	\]
	
	Now let $i, j \in [m]$ with $i < j$ such that $\gcd(f_i, f_j) = 1$. Then
	$\gcd(h_i, h_j) = 1$, because the map:
	\[
		K(t,z_0)[\term{x}] \rightarrow K(t,z_0)[\term{x}], \qquad x_i \mapsto x_i + t^{D_2^i}z_0 \quad (i=1, \dotsc, n)
	\]
	is a $K(t,z_0)$-algebra automorphism. This implies $\res_{z_0}(h_i, h_j) \neq 0$. We have
	\[
		\deg_{\term{x}}\bigl(\res_{z_0}(h_i, h_j)\bigr) \le 2 \delta^2 < D_1 ,
	\]
	therefore the polynomial
	\[
		h_{i,j} := \res_{z_0}\bigl((\mathrm\Psi_D(f_i))(t, z_0, 0, \dotsc, 0), 
		(\mathrm\Psi_D(f_i))(t, z_0, 0, \dotsc, 0)\bigr) \in K[t, z_0]
	\]
	is non-zero (Kronecker substitution). We have
	\[
		\deg_t(h_{i,j}) \le 2\delta^2 \cdot (D_1 + D_1^2 + \dotsb + D_1^n) \le D_1^{n+1}
	\]
	(using $D_1 \ge D_2$) and
	\[
		\sparse(h_{i,j}) \le \max\{\sparse(h_i), \sparse(h_j)\}^{2\delta} \le 2^{2\delta^2}(n+\delta)^{2\delta^2} .
	\]
	Let $B_{2,i,j} \subseteq \PP$ be the set of all primes $p$ satisfying $h_{i,j} \neq 0 \pmod{\langle t^p-1 \rangle_{K[t, z_0]}}$.
	Then $\abs{B_{2,i,j}} < (n+1)2^{2\delta^2}(n+\delta)^{2\delta^2} \log_2D_1$ by Lemma \ref{lem:SparsePIT}.
	Finally, set $B_2 := \bigcup_{i,j} B_{2,i,j}$, where the union is over all $i, j \in [m]$ with $i < j$ such that $\gcd(f_i, f_j) = 1$.
	Then 
	\begin{align*}
		\abs{B_2} &< \tfrac{1}{2} m^2 (n+1)2^{2\delta^2}(n+\delta)^{2\delta^2} \log_2D_1 \\
		&\le \tfrac{1}{2} (ks\delta)^2 (n+1)2^{2\delta^2}(n+\delta)^{2\delta^2} \log_2D_1 .
	\end{align*}
	Ultimately, set $B := B_1 \cup B_2$. Then
	\begin{align*}
		\abs{B} &\le 2 \abs{B_2} < (ks\delta)^2 (n+1)2^{2\delta^2}(n+\delta)^{2\delta^2} \log_2D_1 \\
		&\le (2ksn\delta^2)^{4\delta^2+1} \log_2D_1 .
	\end{align*}
	
	Now pick a suitable prime $p \in \PP \setminus B$ (by Lemma \ref{lem:NumberOfPrimes}). First, let $i \in [m]$. Since $p \notin B_1$, we have 
	$g_i \neq 0 \pmod{\langle t^p-1 \rangle_{K[t]}}$. 
	Define
	\[
		g_i^{(p)} := f_i\bigl(t^{\lfloor D_2^1 \rfloor_p}, \dotsc, t^{\lfloor D_2^n \rfloor_p} \bigr) \in K[t] . 
	\]
	Since $g_i^{(p)} = g_i \neq 0 \pmod{\langle t^p-1 \rangle_{K[t]}}$, we have $g_i^{(p)} \neq 0$. Let $S_{1,i} \subset \ol{K}$
	be the set of all $c \in \ol{K}$ such that $g_i^{(p)}(c) = 0$. Then $\abs{S_{1,i}} \le \deg(g_i^{(p)}) < \delta p$.
	Finally, set $S_1 := S_{1,1} \cup \dotsb \cup S_{1,m}$. Then $\abs{S_1} < m\delta p \le ks \delta^2 p$.
	Now let $i, j \in [m]$ with $i < j$ such that $\gcd(f_i, f_j) = 1$.
	Since $p \notin B_2$, we have $h_{i,j} \neq 0 \pmod{\langle t^p-1 \rangle_{K[t, z_0]}}$. Define
	\[
		h_{i,j}^{(p)} := \res_{z_0}\bigl((\mathrm\Psi_{D,p}(f_i))(t, z_0, 0, \dotsc, 0), 
		(\mathrm\Psi_{D,p}(f_i))(t, z_0, 0, \dotsc, 0)\bigr) \in K[t, z_0] .
	\]
	Since $h_{i,j}^{(p)} = h_{i,j} \neq 0 \pmod{\langle t^p-1 \rangle_{K[t, z_0]}}$, we have $h_{i,j}^{(p)} \neq 0$.
	Let $S_{2,i,j} \subset \ol{K}$ be the set of all $c \in \ol{K}$ such that $h_{i,j}^{(p)}(c, z_0) = 0$. Then
	$\abs{S_{2,i,j}} \le \deg_t(h_{i,j}^{(p)}) < 2 \delta^2 p$. Finally set $S_2 := \bigcup_{i,j} S_{2,i,j}$, where
	the union is over all $i, j \in [m]$ with $i < j$ such that $\gcd(f_i, f_j) = 1$. Then 
	$\abs{S_2} < \tfrac{1}{2} m^2 \cdot 2 \delta^2 p \le \delta^4k^2s^2p$. 
	Ultimately, set $S := S_1 \cup S_2$. Then $\abs{S} < 2\delta^4k^2s^2p$. 
	
	Let $i \in [m]$. Then $\mathrm\Psi_{D,p,c}(f_i)$ is monic in $z_0$ for all $c \in \ol{K} \setminus S$.
	Now let $i, j \in [m]$ with $i < j$ such that $\gcd(f_i, f_j) = 1$. Then
	\begin{multline*}
		\bigl(\res_{z_0}(\mathrm\Psi_{D,p,c}(f_i), \mathrm\Psi_{D,p,c}(f_j))\bigr)(z_0, 0, \dotsc, 0) \\
		= \res_{z_0}\bigl((\mathrm\Psi_{D,p,c}(f_i))(z_0, 0, \dotsc, 0), (\mathrm\Psi_{D,p,c}(f_i))(z_0, 0, \dotsc, 0)\bigr) \\
		= h_{i,j}^{(p)}(c, z_0) \neq 0
	\end{multline*}
	for all $c \in \ol{K} \setminus S$. Thus, $\res_{z_0}(\mathrm\Psi_{D,p,c}(f_i), \mathrm\Psi_{D,p,c}(f_j)) \neq 0$
	and by Lemma \ref{lem:Resultant} it follows that $\gcd(\mathrm\Psi_{D,p,c}(f_i), \mathrm\Psi_{D,p,c}(f_j)) = 1$ 
	for all $c \in \ol{K} \setminus S$. 
\end{proof}

\subsection{Proofs for Sect. \ref{sec:Depth4CircuitsHittingSet}: A hitting set} \label{app:Depth4CircuitsHittingSet}

\noindent {\bf Theorem  \ref{thm:Depth4CircuitsHittingSet}. }
	Assume that $\ch(K) = 0$ or $\ch(K) > \delta^r$. 
	Then $\mathcal{H}_{\delta,k,s}$	is a hitting set for $\spsp_{\delta}(k,s,n)$ circuits.
	It can be constructed in $\poly(\delta rsn)^{\delta^2 kr}$ time.
\begin{proof}[Proof of Theorem \ref{thm:Depth4CircuitsHittingSet}]
	Let $C \in \spsp_{\delta}(k,s,n)$ be a non-zero circuit. First, let us show by a loose estimation that our parameters afford
    $2^k$ applications of Lemmas \ref{lem:VandermondeHomomorphismFaithful} and 
    \ref{lem:PreservingSimplePart} (one for each $\Sp(C_I)$ resp. $C_I$, for all $I \subseteq [k]$). The number of `bad' primes
    by the proofs of these lemmas are at most:
	\begin{align*}
    2^k\cdot &(2nr(n+\delta)^\delta)^{r+1}\log_2D_1 + 2^k\cdot(2ksn\delta^2)^{4\delta^2+1}\log_2D_1 \\
		&< 2^k\cdot(2nr\cdot 2n\delta)^{\delta(r+1)}\log_2D_1 + 2^k\cdot(2ksn\delta^2)^{4\delta^2+1}\log_2D_1 \\
		&< 2^k\cdot(2nr\delta)^{2\delta(r+1)}\log_2D_1 + 2^k\cdot(2ksn\delta^2)^{4\delta^2+1}\log_2D_1 \\
		&< 2^{k+1}\cdot(2krsn\delta^2)^{4\delta^2+2\delta r}\log_2D_1 .
	\end{align*}
	Thus, the set $[p_{\max}]$ would have a `good' prime $p$ (by Lemma \ref{lem:NumberOfPrimes}). Next comes the estimate 
	on the number of `bad' $c$:
	\[
		2^k\delta rp + 2^k\cdot(2\delta^4k^2s^2p) < 	2^{k+2}k^2rs^2\delta^4 p .
 	\]
	Thus, Lemma \ref{lem:VandermondeHomomorphismFaithful}
	and Lemma \ref{lem:PreservingSimplePart} imply that there exist a prime $p \in [p_{\max}]$ and an element 
	$c \in H_1$ such that, for all $I \subseteq [k]$, we have
	\begin{enumerate}
		\item $\mathrm\Psi_{D,p,c}(\simple(C_I)) = \simple(\mathrm\Psi_{D,p,c}(C_I))$, and
		\item $\mathrm\Psi_{D,p,c}$ is faithful to some subset $\{f_1, \dotsc, f_m\} \subseteq \Sp(\simple(C_I))$
			of transcendence degree $\min\{\rk(\simple(C_I)), r\}$.
	\end{enumerate}
	Hence, by Lemma \ref{lem:RankBasedBlackBoxAlgo}, $\mathrm\Psi_{D,p,c}(C)$ is a non-zero circuit with at most $r+1$ variables 
	and of degree at most $\delta s$. Now the first assertion follows from
	Lemma \ref{lem:CombinatorialNullstellensatz}. The second assertion is obvious from the construction. 
\end{proof}

\noindent {\bf Corollary \ref{cor:Depth4CircuitsHittingSetTopFanIn2}.}
	Let $K$ be of arbitrary characteristic. Then $\mathcal{H}_{\delta, 2, s}$ is a hitting set for 
	$\spsp_{\delta}(2,s,n)$ circuits. It can be constructed in $\poly(\delta sn)^{\delta^2}$ time.
\begin{proof}[Proof of Corollary \ref{cor:Depth4CircuitsHittingSetTopFanIn2}]
	First observe $R_{\delta}(2,s) = 1$. Since $\mathrm\Psi$ sends non-constant sparse polynomials
	of a circuit to non-constant polynomials (see the proof of Lemma \ref{lem:PreservingSimplePart}), 
	it is faithful to sets of transcendence degree $1$. Hence we do not need to invoke
	Lemma \ref{lem:VandermondeHomomorphismFaithful} (where the dependence on the characteristic
	came from). 
\end{proof}
	
\end{appendix}
\end{document}